\documentclass[aip,reprint,jcp,amsmath,amssymb]{revtex4-1}
\usepackage[T1]{fontenc}
\usepackage{amsthm}
\usepackage{amsthm}
\usepackage{enumitem}
\usepackage{color}
\usepackage{graphicx}
\usepackage{hyperref}
\usepackage[capitalise]{cleveref}
\usepackage[english]{babel}
\usepackage{gensymb}
\def\W{\mathcal{W}}
\def\dt{\Delta t}
\def\manq{\Sigma}
\def\ds{\displaystyle}
\def\E{\mathbb{E}}
\def\R{\mathbb{R}}
\def\ph{\varphi}
\def\one{\mathbf{1}}
\def\dps{\displaystyle}
\newcommand{\pare}[1]{ \left(#1\right) }

\newtheorem{theorem}{Theorem}[section]
\newtheorem{lem}[theorem]{Lemma}
\theoremstyle{remark}
\newtheorem{rem}[theorem]{Remark}
\numberwithin{equation}{section}
\crefname{rem}{Remark}{Remarks}
\Crefname{rem}{Remark}{Remarks}
\usepackage[normalem]{ulem}
\begin{document}

\title{Efficient Monte-Carlo sampling of metastable systems using non-local collective variable updates} 

\author{Christoph Schönle}
\affiliation{CMAP, CNRS, École polytechnique, Institut Polytechnique de Paris, 91120 Palaiseau, France}
\affiliation{Laboratoire de Physique de l’École normale supérieure ENS, Université PSL, CNRS, Sorbonne Université, Université de Paris, 75005 Paris, France}

\author{Davide Carbone}
\affiliation{Laboratoire de Physique de l’École normale supérieure ENS, Université PSL, CNRS, Sorbonne Université, Université de Paris, 75005 Paris, France}

\author{Marylou Gabrié}
\affiliation{Laboratoire de Physique de l’École normale supérieure ENS, Université PSL, CNRS, Sorbonne Université, Université de Paris, 75005 Paris, France}

\author{Tony Lelièvre}
\affiliation{CERMICS, CNRS, ENPC, Institut Polytechnique de Paris, Marne-la-Vallée, France}
\affiliation{MATHERIALS team-project, Inria Paris, France}

\author{Gabriel Stoltz}
\affiliation{CERMICS, CNRS, ENPC, Institut Polytechnique de Paris, Marne-la-Vallée, France}
\affiliation{MATHERIALS team-project, Inria Paris, France}

\date{\today}

\begin{abstract}
Monte-Carlo simulations are widely used to simulate complex molecular systems, but standard approaches suffer from metastability. Lately, the use of non-local proposal updates in a collective-variable (CV) space has been proposed in several works. Here, we generalize these approaches and explicitly spell out an algorithm for non-linear CVs and underdamped Langevin dynamics. We prove reversibility of the resulting scheme and demonstrate its performance on several numerical examples, observing a substantial performance increase compared to methods based on overdamped Langevin dynamics as considered previously. Advances in generative machine-learning-based proposal samplers now enable efficient sampling in CV spaces of intermediate dimensionality (tens to hundreds of variables), and our results extend their applicability toward more realistic molecular systems.
\end{abstract}

\maketitle 
\section{Introduction}

Monte-Carlo simulations promise to give access to the equilibrium properties of physical systems even when the numerical simulation of their relaxation dynamics is prohibitively expensive. In practice, however, standard Markov Chain Monte Carlo methods (MCMC) rely on physically-inspired local moves and struggle to equilibrate across metastable states, as does molecular dynamics.

Enhanced sampling algorithms go beyond the naive Markov Chain Monte Carlo approach to tackle metastability, see Ref.~\onlinecite{Henin_Lelievre_Shirts_Valsson_Delemotte_2022} for a review in the context of molecular dynamics. A first class of methods originated from the observation that the metastability is removed at sufficiently high temperatures, which led to the development of algorithms like parallel tempering, annealed importance sampling, and sequential Monte Carlo \cite{swendsenReplicaMonteCarlo1986,nealAnnealedImportanceSampling2001,delmoralSequentialMonteCarlo2006,woodard_conditions_2009,syed_non-reversible_2022}. While known to work on a wide range of problems, they spend a substantial amount of computation time at temperatures not of direct interest and generally require a careful choice of parameters (e.g. for parallel tempering the temperature ladder and the frequency of exchange attempts).

Another important category of algorithms makes use of a collective variable (CV), a lower-dimensional representation of the system that is supposed to capture the metastability of the system. A `good' CV that resolves the metastability can be used for biasing simulations to bridge modes, using enhanced sampling algorithms such as umbrella sampling\cite{torrie-valleau-77} and related histogram methods\cite{KBSKR92,SC08}, thermodynamic integration\cite{kirkwood-35} or free energy adaptive biasing methods\cite{darve-pohorille-01,henin-chipot-04,laio-parrinello-02,barducci-bussi-parrinello-08}. However, identifying a suitable CV in low dimension is by no means trivial and thus poses the main limitation to this class of methods.

In recent years, another approach has emerged, which is exploiting generative machine learning models. Among the various generative models, normalizing flows\cite{papamakarios2021} are particularly well-suited for the purposes of MCMC sampling since their tractable probability density allows to employ them as a non-local proposal density of physical configurations, where any imperfections in training are corrected for via a Metropolis-Hastings accept/reject scheme or self-normalized importance sampling\cite{noeBoltzmannGeneratorsSampling2019,greniouxSamplingApproximateTransport2023b}. Thanks to the impressive flexibility of neural networks, proofs of concept have been made in different domains ranging from molecular systems to field theories in quantum chromodynamics\cite{liNeuralNetworkRenormalization2018,noeBoltzmannGeneratorsSampling2019,gabrieAdaptiveMonteCarlo2022,invernizzi_skipping_2022,abbott_sampling_2022}. However, the limiting factor to scale these approaches appears to be the ability of the generative model to be precise enough either when the dimension of the problem becomes of the order of a few thousands or when the distribution has singularities\cite{deldebbioEfficientModellingTrivializing2021,greniouxSamplingApproximateTransport2023b,schonleOptimizingMarkovChain2023,abbottAspectsScalingScalability2022}.Other recent approaches make use for example of transformers~\cite{bera_accurate_2025} and diffusion models~\cite{lewisScalableEmulationProtein2025b,bera_how_2025}, but these do not give easy access to the underlying probability density and are therefore less suitable for MCMC sampling than normalizing flows.

A promising approach is therefore to combine the strength of CV-enhanced sampling and generative models (normalizing flows) with a CV of `intermediate' dimension (from a dozen to a few hundred degrees of freedom). Within this dimensionality, identifying a CV such that it resolves the metastability will be much easier than finding a one- to three-dimensional representation required by most classical CV-based approaches. Conversely, a normalizing flow in an `intermediate' dimension can be trained to much higher accuracy than if one attempted to learn the distribution of the entire system. Since the overall goal is still sampling full configurations of the system, the question arises of how to employ such a proposal sampler in CV space to obtain an unbiased sampler for the full distribution. This has been addressed in multiple works\cite{athenesComputationChemicalPotential2002,nilmeier2011nonequilibrium,chen_enhanced_2015,chen2015generalized,neal2005taking}, in particular under the names of ``Hybrid Nonequilibrium Molecular
Dynamics'' (HNMD)\cite{chen_enhanced_2015,chen2015generalized} and ``Nonequilibrium candidate Monte Carlo'' (NCMC) \cite{nilmeier2011nonequilibrium} in the context of molecular dynamics, but also in a Bayesian setting\cite{karagiannis_annealed_2013}. The central idea is always to drive the system from one value of the CV to another, a procedure referred to as steering in the following, and then accept or reject the resulting proposal in the full-space. A similar procedure has also been used in the simulation of glassy system with shear deformations\cite{das_annealing_2022}. In a previous article\cite{SCHONLE2025113806}, we described such an algorithm with driven dynamics in continuous time
and identified various time-discretizations for which we proved reversibility with respect to the equilibrium measure. Notably, we 
also related this to the Jarzynski and the Jarzynski--Crooks equalities
\cite{jarzynski-97,Crooks99}, clarifying the algorithm design. In a parallel work\cite{tamagnoneCoarseGrainedMolecularDynamics2024}, one version of these algorithms was applied to a toy molecular system, using a normalizing flow as the CV sampler and including an adaptive protocol to train it alongside the MCMC procedure.

However, both these works considered only overdamped Langevin dynamics for the steering and required the CV to be a subset of the Euclidean coordinates giving the conformation of the system. In the present article, we generalize these works by presenting an algorithm for a non-linear collective variable and using general underdamped Langevin dynamics. While in spirit this algorithm already fits into the paradigm of some of the previous works like NCMC and HNMD, we present here an algorithm that is readily implemented and treat all technical details arising from the non-linear CV and general Langevin dynamics. We explicitly prove that the associated Markov chains are reversible with respect to the target measure and present its application to several example systems of increasing complexity. In all of them, we see that the deterministic (Hamiltonian) steering dynamics exhibit the best performance, with an improvement of up to two orders of magnitude compared to overdamped dynamics.

The main part of the article divides into two sections: We explain all the theoretical aspects of the problem and the algorithm in \cref{sec:sampling_algo} and go on to demonstrate its properties through several numerical examples in \cref{sec:numerics}. Within \cref{sec:sampling_algo}, the sampling problem is introduced in \cref{subsec:target}, followed by a detailed explanation of the algorithm in \cref{subsec:algorithm}. We state our main theoretical result on the reversibility of the algorithm in \cref{subsec:reversibility} and conclude in \cref{subsec:Parameterization} with a normalized version of the algorithm and comments on the parameterization. In \cref{sec:numerics}, we analyze the performance of the algorithm by means of numerical simulations and present results on four different model systems of increasing complexity in \cref{subsec:gauss_tunnel,subsec:phi4,subsec:dimer,subsec:polymer}. We conclude the article and give an outlook on future directions in \cref{sec:conclusion}.

\section{Sampling with a non-linear collective variable}
\label{sec:sampling_algo}
\begin{figure*}
    \includegraphics{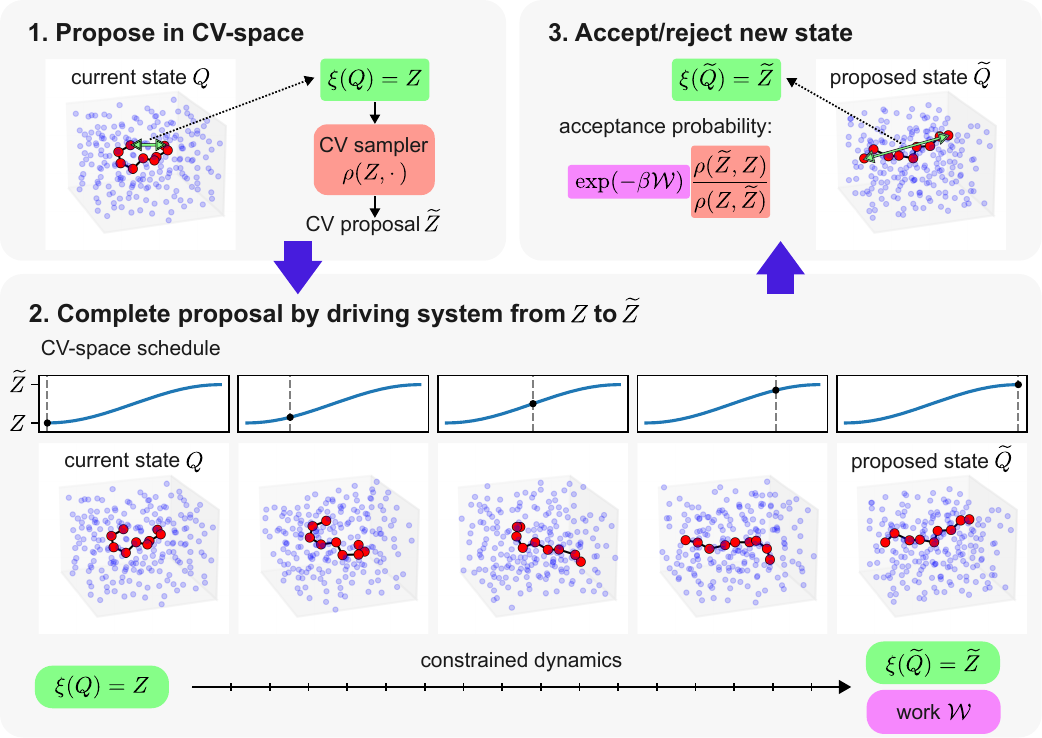}
    \caption{Schematic representation of the algorithm for the example of a polymer, here with the end-to-end distance as the CV.}
    \label{fig:overviewAlgo}
\end{figure*}
\subsection{Target distribution and collective variable}
\label{subsec:target}
We begin this section by outlining the sampling problem at hand.
Our goal is to sample from the Boltzmann--Gibbs probability measure 
\begin{align}
    \nu(\mathrm{d}q) = \mathcal{Z}_q^{-1} e^{-\beta V(q)}\mathrm{d}q
    \label{eq:boltzmanntarget}
\end{align} over positions $q\in\mathbb{R}^d$ of dimension $d$, with the energy function $V:\mathbb{R}^d\rightarrow\mathbb{R}$, the inverse temperature $\beta=(k_BT)^{-1} > 0$ and the normalizing constant $\mathcal{Z}_q$. Let us assume that a collective variable (CV) $\xi:\mathbb{R}^d \rightarrow \mathbb{R}^\ell$ with $\ell<d$ has been identified and we have a `good' sampler $\rho$ (close to the target distribution in CV space) to propose moves in CV space. As mentioned in the introduction, the question of how to use these two ingredients to construct an unbiased MCMC sampler in the full space $\mathbb{R}^d$ has previously been addressed in literature\cite{athenesComputationChemicalPotential2002,nilmeier2011nonequilibrium,chen_enhanced_2015,chen2015generalized,neal2005taking}, but it is not obvious how to employ these schemes if the map $\xi$ is non-linear. Notably, the geometry of the level sets $\Sigma(z)$ defined as
\begin{align}
\label{eq:Sigmaz}
    \Sigma(z)  = \left\{q\in\mathbb{R}^d \,\Big | \, \xi(q) = z \right\}
\end{align}
must be taken into consideration. We present an algorithm which explicitly treats the case of a non-linear CV in the following section. Much of the notation throughout this paper is based on Refs.~\onlinecite{lelievre-rousset-stoltz-book-10,lelievre_langevin_2012} and we refer the reader to those for an in-depth introduction on the  mathematical concepts used here.

Even though we are only interested in sampling $\nu(\mathrm{d}q)$, it will turn out to be helpful to introduce also an auxiliary momentum variable $p$ and define the Boltzmann--Gibbs measure in phase space $(q,p)$ by 
\begin{align}
    \mu(\mathrm{d} q \,\mathrm{d} p) = \mathcal{Z}^{-1} e^{-\beta H(q,p)} \mathrm{d} q \, \mathrm{d} p
\end{align} where $\mathcal{Z}$ is the normalizing constant, ${H(q,p) = \frac{1}{2} p^\top M^{-1} p + V(q)}$ is the Hamiltonian, and $M\in \mathbb{R}^{d\times d}$ is a symmetric positive definite mass matrix.
\subsection{The algorithm}
\label{subsec:algorithm}
To explain the algorithm, we first present its overall structure in \cref{subsubsec:algo}, followed by the necessary ingredients and considerations in \cref{subsubsec:schedule,subsubsec:momentum_init,subsubsec:Rattle}. 

\subsubsection{Overall structure}
\label{subsubsec:algo}
The algorithm builds an MCMC chain $(Q_n)_{n\geq 0}$ in position space by iterating over $n \in \mathbb{N}$ the following steps.
\noindent
Starting from a state $Q_n$ with collective variable value $Z_n := \xi(Q_n)$, generate $Q_{n+1}$ as follows:
\begin{enumerate}[label=\textbf{Step \arabic*}:, leftmargin=*]
    \item Propose a move in the collective variable space:
      \[
      \widetilde{Z}_{n+1} \sim \rho(Z_n,z') \, {\rm d}z',
      \]
    where $\rho(z,z') \in \R^\ell \times \R^\ell \to \R_+$ denotes a Markov kernel with probability density $\rho(z,z')$ to jump to $z'$ starting at $z$. 
    \item Generate a new $\widetilde{Q}_{n+1} \in \mathbb{R}^d$ that satisfies $\xi(\widetilde{Q}_{n+1}) = \widetilde{Z}_{n+1}$ by driving the system in CV space from $Z_n$ to $\widetilde{Z}_{n+1}$. To this end, 
    \begin{itemize}
        \item[(i)]  specify a schedule for the positions and velocities in CV space to fix the transition path (\cref{subsubsec:schedule}),
        \item[(ii)] initialize the position $q^0 = Q_n$ and sample an initial auxiliary momentum $p^0$ (\cref{subsubsec:momentum_init}),
        \item[(iii)] use this to follow (Langevin) dynamics constrained in CV space by the schedule (\cref{subsubsec:Rattle}),
    \end{itemize}  
    all of these steps to be specified below.
    These dynamics yield a final state $\widetilde{Q}_{n+1}$ and a quantity $\mathcal{W}_{n+1}$ that is identified as a work. The final momentum is discarded.
    \item Draw a random variable $U_{n+1}$ with uniform law over $[0,1]$ and
    \begin{itemize}
        \item If $U_{n+1} \le \exp(-\beta \mathcal{W}_{n+1})) \displaystyle \frac{\rho(\widetilde{Z}_{n+1}, Z_n)}{\rho(Z_n,\widetilde{Z}_{n+1})}$,\\
        accept the proposal: $Q_{n+1}=\widetilde{Q}_{n+1}$.
        \item Otherwise, reject the proposal: \\$Q_{n+1} = Q_n$.
    \end{itemize}
\end{enumerate}

The three steps of the algorithm are schematically illustrated in \cref{fig:overviewAlgo}.

In principle, any probability density $\rho(z, z')$ can be chosen to generate the proposal moves in CV space in \textbf{Step~1}, since the accept/reject criterion in \textbf{Step~3} always ensures unbiasedness, as will be proven below. This choice will, however, affect the acceptance probability and thus the performance of the algorithm. In our numerical examples, we always use an independent sampler in the CV space, i.e. $\rho(z,z')$ is a function of~$z'$ only. A natural choice is (an approximation of) the marginal density~$\rho(z')$ of the target measure in CV space -- that is the image measure of the target distribution \eqref{eq:boltzmanntarget} under the CV map $\xi$. For a CV of intermediate dimension, such a proposal sampler could for example be learned by a normalizing flow.

We explain the various ingredients required for \textbf{Step~2} in the next three subsections. The idea to sample a system by following some dynamics followed by an accept/reject step that is governed by a work is closely related to the Jarzynski and the Jarzynski--Crooks equalities
\cite{jarzynski-97,Crooks99}.
\subsubsection{Schedule}
\label{subsubsec:schedule}
To specify the constrained dynamics under which the system evolves from the value of the CV at the previous iteration $Z_n$ to the newly proposed value $\widetilde{Z}_{n+1}$, choose first the number of intermediate steps $K_T$ associated with this jump. With a fixed time discretization $\Delta t$, this sets the physical transition time $T=K_T \Delta t$. In practice, we will choose $K_T$ as a function of the Euclidean distance $|\widetilde{Z}_{n+1} - Z_n|$ of the proposed jump. 

Introduce a schedule $(z^{(Z_n, \widetilde{Z}_{n+1})}(t_k))_{0\le k \le  K_{T}}$ with $t_k = k\Delta t$ whose endpoints are
\begin{align}
    z^{(Z_n, \widetilde{Z}_{n+1})}(t_0)&=Z_n\\
z^{(Z_n, \widetilde{Z}_{n+1})}(t_{K_T})&=\widetilde{Z}_{n+1}
\end{align}
as well as a velocity schedule
$(v_z^{(Z_n, \widetilde{Z}_{n+1})}(t_k))_{0\le k \le  K_{T}}$ with zero initial and final velocity, 
\begin{align}
    v^{(Z_n, \widetilde{Z}_{n+1})}_z(t_0)  = v_z^{(Z_n, \widetilde{Z}_{n+1})}(t_{K_T}) =0.
    \label{eq:zerovelocity_condition}
\end{align} 
We require that the schedules for position and velocity satisfy the following reversibility property with the switching of the endpoints: for any $(Z, \widetilde{Z})$, we need that 
\begin{align}
    z^{(Z, \widetilde{Z})}(t_k) &= z^{(\widetilde{Z}, Z)}(t_{K_T-k})\label{eq:rev_positions}\\
    v_z^{(Z, \widetilde{Z})}(t_k) &= -v_z^{(\widetilde{Z}, Z)}(t_{K_T-k}) \label{eq:rev_velocities},
\end{align}
 for all $k \in \{0, \ldots, K_T\}$. In practice, it is advisable for numerical reasons to choose the velocity schedule in accordance with the position schedule (see \cref{rem:velocity_schedule_choice} in the appendix). Note, however, that the requirement \eqref{eq:rev_velocities} precludes the simplest choice for the velocity of $v_z(t_k) = (z(t_{k+1}) - z(t_k))/\Delta t$. A simple choice that admits a continuous-in-time limit when $\Delta t$ goes to zero is instead a schedule that connects the endpoints on a straight line with a continuously differentiable schedule function $f:[0,1]\rightarrow [0,1]$:
\begin{align}
\label{eq:schedule_function_z}
    z^{(Z, \widetilde{Z})}(t_k) &= Z + (\widetilde{Z} - Z) f(k/K_T) \\
    v^{(Z, \widetilde{Z})}_z(t_k) &= (\widetilde{Z} - Z) \frac{f'(k/K_T)}{K_T \Delta t}.
    \label{eq:schedule_function_vz}
\end{align}
The conditions on the schedule translate to conditions on $f$ as $f(0) = 0$, $f(1)=1$ and $f'(0) = f'(1)=0$, which are fulfilled for example by 
\begin{align}
    f(\tau)=(1-\cos(\tau\pi))/2.
    \label{eq:schedule_cos}
\end{align}
This is one simple choice to parameterize the schedule, but many others are possible as long as the requirements stated above are fulfilled.

The choice of transition path in CV space as well as the time schedule to follow it determines how far from equilibrium the dynamics will take the system, which affects the acceptance ratio and thus the overall numerical performance (theoretically, the highest possible acceptance will be achieved by driving infinitely slowly, but at the expense of longer computational times to integrate the steered dynamics). In our examples, we stick to the simple prescription explained above, where the transition path is a straight line in CV space, and only optimize the overall speed with which that path is followed.
See \cref{rem:velocity_schedule_choice} in \cref{subsubsec:appendix_forw} for further comments on the velocity schedule. The assumptions~\eqref{eq:zerovelocity_condition} and~\eqref{eq:rev_velocities} play a crucial role in the proof of the reversibility of the Markov chain $(Q_n)_{n \ge 0}$ and we demonstrate in \cref{subsec:AppNonlinGauss} how their violation can introduce a bias. Note that for some choices of CV, the condition \eqref{eq:zerovelocity_condition} can be somewhat relaxed, see \cref{rem:constant_fixman}.
\begin{rem}
    Even though this is not explicitly indicated, the total time $T$ can also depend on $(Z,\widetilde{Z})$ in which case one requires $T^{(Z,\widetilde{Z})} = T^{(\widetilde{Z},Z)}$.
\end{rem}

\subsubsection{Initialization of the momentum}
\label{subsubsec:momentum_init}
In each iteration of the algorithm, a fresh momentum is sampled in \textbf{Step 2} to initialize the constrained dynamics. To agree with condition \eqref{eq:zerovelocity_condition}, this momentum must correspond to zero velocity in $\xi$-direction (i.e. in CV space) and match with the underlying equilibrium canonical distribution orthogonally to $\xi$. We explain here how to achieve this.
For a given position $q$ and momentum $p$, the velocity $v_\xi$ is defined as
\begin{align}
    v_\xi(q,p)= \nabla\xi(q)^\top M^{-1} p\in \mathbb{R}^\ell
\end{align}
where $\nabla \xi(q)=\left(\nabla \xi_1(q), \ldots, \nabla \xi_{\ell}(q)\right) \in \R^{d \times \ell}$.
The cotangent space of zero velocity associated with a position $q\in \Sigma(z)$ is defined as
\begin{align}\label{eq:cotangent_zero_speed}
    T_q^*\Sigma(z) = \left\{p\in \mathbb{R}^d\, \middle| \, v_{\xi}(q,p)= 0\right\},
\end{align}
and the orthogonal projection on $T_q^*\Sigma(z)$ (for the scalar product induced by $M^{-1}$) is given by:
    \begin{align}\label{eq:P_M}
        P_M(q) = \operatorname{Id}_d - \nabla\xi(q) G_M^{-1}(q) \nabla\xi(q)^\top M^{-1}.
    \end{align}
Here, $G_M$ is the Gram tensor defined as
\begin{equation}\label{eq:G_M}
G_M(q) = \nabla\xi(q)^\top M^{-1} \nabla\xi(q).
\end{equation}
Throughout this paper, we assume that $G_M$ is invertible on the submanifolds $\Sigma(z)$ (for all $z$).
It is easy to check that, for a given $q\in \Sigma(z)$, one has $P_M(q) p\in T_q^*\Sigma(z)$ for any $p$ and that $P_M(q)^2=P_M(q)$ and $M^{-1}P_M(q)=P_M(q)^\top M^{-1}$ and thus $P_M$ is indeed the orthogonal projection, as stated above. 

We therefore initialize the momentum $p^0$ by sampling an unconstrained momentum $\widetilde{p}^0$ from a centered normal distribution with covariance matrix $\frac{1}{\beta} M$ and projecting it onto the subspace of zero velocity:
\begin{align*}
    p^0 = P_M(q^0) \widetilde{p}^0.
\end{align*}

\subsubsection{Constrained dynamics and work}
\label{subsubsec:Rattle}
We now specify how the system is driven under constrained dynamics to follow the chosen CV and velocity schedules. Special care needs to be taken here regarding the geometry of the level sets $\Sigma(z)$. To account for the change in phase space volume due to the rigid constraints, we need to follow the dynamics with a modified potential energy $\widetilde{V}(q) = V(q) + V_\mathrm{fix}(q)$, with the so-called Fixman term\cite{Fixman1978}:
\begin{align}
V_\mathrm{fix}(q) = \frac{1}{2\beta}\log \left(\det G_M(q)\right),
\label{eq:deffixman}
\end{align}
(see also the reversibility proof in \cref{sec:proof_reversibility} and \cref{rem:free_energy_relations_nofixman,rem:proof_no_fixma_velocities} for the technical reason why it is necessary).
We numerically demonstrate the necessity of this modification in \cref{subsec:AppNonlinGauss}. Note that the tilde on $\widetilde{V}$ and $\widetilde{H}$ indicates here the modification by the Fixman term, while we also use it on $\widetilde{Q}_{n+1}$ and $\widetilde{Z}_{n+1}$ to designate proposal states in the MCMC chain.

Starting from $(q^0, p^0)$, a path $\left((q^k,p^k)_{0 \le k \le K_T}\right)$ is generated by discretizing in time a constrained Langevin dynamics, with a splitting scheme of the form (midpoint Euler)-(Verlet)-(midpoint Euler) (see for example Ref.~\onlinecite{lelievre_langevin_2012}). The first and last parts (the thermostat parts) insert fluctuations and the middle part realizes Hamiltonian dynamics. For the latter, we use a RATTLE discretization scheme with constraints, slightly different to the standard version since we consider constraints which evolve in time (see also~\cref{rem:velocity_schedule_choice} in the Appendix for further discussion). For the sake of readability, we drop the superscript indicating the endpoints of the schedules $z(t)$ and $v_z(t)$ in this subsection. The scheme writes as follows: For $0 \leq k \leq K_T-1$, update the positions and momenta as
\begin{align}
 p^{k+1/4} =   p^{k} &-\frac{\dt}{4} \gamma_P(q^k) M^{-1}(p^{k+1/4}+p^{k}) \notag \\
       &+ \sqrt{ \frac{\dt}{2}} \sigma_P(q^k) {\mathcal G}^k 
     \label{eq:flucdissconstjarz1}
\end{align}
\begin{align}
  \begin{cases}
    \dps  p^{k+1/2} = p^{k+1/4} - \displaystyle{\frac{\dt}{2} \nabla \widetilde{V} (q^{k})} + \nabla \xi(q^k) \lambda^{k+1/2} \\[6pt]
    \dps  q^{k+1} = q^{k} + \dt \ M^{-1} p^{k+1/2}  \\[6pt]
    \dps   \xi(q^{k+1})  = z(t_{k+1}) \quad\quad (C_q) \\[6pt]
    \dps  p^{k+3/4} = p^{k+1/2} - \displaystyle{\frac{\dt}{2} \nabla \widetilde{V} (q^{k+1})} +\nabla \xi(q^{k+1}) \lambda^{k+3/4}\\[6pt]
    \dps   \nabla \xi (q^{k+1})^{T} M^{-1} p^{k+3/4} = v_z(t_{k+1})\quad\quad (C_p)
\end{cases}\label{eq:Verletconstswitched}
\end{align}
\begin{align}
 p^{k+1} = p^{k+3/4} &-\frac{\dt}{4} \gamma_P(q^{k+1}) M^{-1}(p^{k+3/4}+p^{k+1}) \notag \\
 &+ \sqrt{\frac{\dt}{2}} \sigma_P(q^{k+1}) {\mathcal G}^{k+1/2}
      \label{eq:flucdissconstjarz2}
\end{align}
where $({\mathcal G}^k)$ and $({\mathcal G}^{k+1/2})$ are independent sequences of i.i.d. centered Gaussian random vectors with covariance matrix~$\mathrm{Id}_{d}$. Note that the constraint in~$(C_q)$ is time-dependent, and thus the constraint in~$(C_p)$ is non zero, in contrast to the standard RATTLE scheme. The Lagrange multiplier $\lambda^{k+1/2}$ (resp. $\lambda^{k+3/4}$) is still chosen so that the constraint $(C_q)$ (resp. $(C_p)$) is satisfied. For $\lambda^{k+3/4}$, this leads to a linear equation with an analytical solution (as for the usual RATTLE scheme), whereas the solution for $\lambda^{k+1/2}$ will in general be obtained numerically using Newton's methods. For further details on this point, in particular regarding uniqueness of the solution, see \cref{rem:existence_lagrange_multip} in the Appendix. Note that there is an equivalent formulation of \eqref{eq:flucdissconstjarz1} and \eqref{eq:flucdissconstjarz2} that is more suitable for numerical implementation, which we provide in \cref{app:numericalformulation} for completeness. Besides, let us emphasize that the matrices $\gamma_P(q^k) = P_M(q) \gamma P_M(q)^\top$ and $\sigma_P(q^k)=P_M(q)\sigma$ satisfy the fluctuation-dissipation identity
\begin{equation}\label{eq:fluct_dissip}
    \forall q\in \mathbb{R}^d,\qquad \sigma_P(q) \sigma_P(q)^\top=\frac{2}{\beta} \gamma_P(q),
\end{equation}
ensured by defining the scalar diffusion parameter $\sigma$ in terms of the scalar damping parameter $\gamma > 0$ as
\begin{equation}\label{eq:scalar_fluct_dissip}
\sigma^2=\frac{2 \gamma}{\beta}.
\end{equation}
We stick here for simplicity to scalar damping and diffusion parameters, but generalization to matrices is straightforward.

Assuming that the position and velocity schedules are compatible in the continuous-time limit $\Delta t \to 0$, i.e. $v_z(t) = \dot{z}(t)$, \cref{eq:flucdissconstjarz1,eq:Verletconstswitched,eq:flucdissconstjarz2,} are discretizations of Langevin dynamics with constraints  $(\xi(q(t)),v_\xi(q(t),p(t)))=(z(t),\dot{z}(t))$, see Ref.~\onlinecite{lelievre_langevin_2012}.

These dynamics yield a proposal move in the full configuration space $\widetilde{Q}_{n+1}=q^{K_T}$ as well as a work $\mathcal{W}_{n+1} := \mathcal{W}^{K_T}$, which we use to compute the acceptance probability in \textbf{Step 3}. This work is defined as
    \begin{equation}
        \left\{
        \begin{aligned}
                \mathcal{W}^0 &= 0\\
                \mathcal{W}^{k+1} &= \mathcal{W}^k + \widetilde{H}(q^{k+1},p^{k+3/4}) - \widetilde{H}(q^k,p^{k+1/4}),
        \end{aligned}
        \right.
        \label{eq:work_definition}
        \end{equation}
for $k=0\ldots K_T-1$, where $\widetilde{H}$ denotes the Hamiltonian that is modified by the Fixman term \eqref{eq:deffixman}, namely ${\widetilde{H}(q,p) = \widetilde{V}(q) + \frac{1}{2}p^\top M^{-1} p}$.

For $\gamma=0$, the dynamics reduce to the deterministic part \eqref{eq:Verletconstswitched} (driven Hamiltonian dynamics) since \eqref{eq:flucdissconstjarz1} and \eqref{eq:flucdissconstjarz2} leave the system unchanged. The only remaining source of randomness is thus the initialization of the momentum $p^0$ in \textbf{Step 2}. Furthermore, since in this case $p^{k+1/4}=p^k$ and $p^{k+3/4}=p^{k+1}$, the work simplifies to a telescopic sum and just depends on the endpoints of the path: \[\mathcal{W}^{K_T} = \widetilde{H}(q^{K_T},p^{K_T}) - \widetilde{H}(q^0,p^{0}).\]

The idea of resampling the momentum at each iteration of the algorithm to build an MCMC sampler in position space is akin to Hamiltonian Monte-Carlo (HMC)~\cite{DUANE1987216} and for $\gamma=0$, this analogy is particularly apparent. Unlike HMC, however, where the total energy is conserved up to a time discretization error, our algorithm does not conserve the energy due to the time-evolving constraints in CV space. Rather, the accumulated work is used in \textbf{Step 3} of the algorithm to guarantee the reversibility of the Markov kernel.

\begin{rem}
\label{rem:constant_fixman}
When the Gram tensor $G_M(q)$ does not depend on $q$ (an example is given in \cref{subsec:dimer}), the Fixman term just adds an overall shift to the energy $V$ and need not be considered. In this case, the requirement \eqref{eq:zerovelocity_condition} on the initial and final velocity can also be relaxed to $v_z(t_0) = v_z(t_{K_T})$ and can be non-zero, assuming that the initial momentum $p^0$ is still sampled from the appropriate subspace (see  \cref{rem:free_energy_relations_nofixman,rem:proof_no_fixma_velocities}). This is true in particular for the case of a linear CV, which, after a change of basis, can always be seen as representing a subset of $\ell$ degrees of freedom. This is the standard setting treated in previous works\cite{athenesComputationChemicalPotential2002,nilmeier2011nonequilibrium,chen_enhanced_2015,chen2015generalized,neal2005taking}, where some degrees of freedom are driven through a non-equilibrium protocol. We spell out the dynamics for the simplified setting of the linear CV in \cref{sec:linearCV}, since it allows for an intuitive understanding as an alchemical transition with the CV as the alchemical parameter.
\end{rem}
The problem of constructing a steered path can be understood as transporting samples from the constrained target measure at the initial CV value $Z_n$ to the final one $\widetilde{Z}_{n+1}$. Our method is based on a stochastic differential equation (with the auxiliary variable of the momentum), but other dynamics are possible, notably deterministic transformations\cite{chen_enhanced_2015} or a combination of stochastic and deterministic steps, where the underlying SDE is modified by a deterministic `escorting' term\cite{vaikuntanathan_escorted_2011}. The latter can also be understood as an instance of a stochastic normalizing flow\cite{caselle_stochastic_2022}. These approaches are not easy to adapt to our context of a nonlinear CV, however, since they require to transport samples supported on the submanifolds $\Sigma_z$.
\subsection{Reversibility of the algorithm}
\label{subsec:reversibility}
Our main theoretical result is the following theorem:
\begin{theorem}
\label{theo:reversbility}
The MCMC algorithm presented in \cref{subsec:algorithm} is reversible with respect to the target measure $\nu(\mathrm{d}q)$.
\end{theorem}
We prove this statement in \cref{sec:appendix_reversibility}. The basic ingredient of the proof is to introduce a backward version of the dynamics (associated to the forward dynamics) and use a version of the Jarzynski--Crooks equality\cite{Crooks99}.

As a consequence of \cref{theo:reversbility}, the MCMC algorithm admits $\nu$ as an invariant measure and, up to checking irreducibility which has to be done on a case-by-case basis, empirical averages over trajectories $(Q_n)_{n \ge 0}$ converge to averages with respect to $\nu$.
\begin{rem}
    In \cref{subsec:AppNonlinGauss}, we numerically demonstrate how the omission of the Fixman term or an inappropriate steering schedule can lead to a bias in the simulations, underlining that these two ingredients are indeed necessary for an unbiased algorithm.
\end{rem}
\subsection{A normalized version of the algorithm}
\label{subsec:Parameterization}
The overall algorithm depends on many parameters: $\Delta t$, $\gamma$, $\sigma$, $M$, the choice of the schedule $z(t)$ and $v_z(t)$, and the transition time $T$. In order to reduce the number of parameters, we will use the following normalized version in the numerical examples: Assuming that the mass matrix $M=M\,\mathrm{Id}$ is isotropic, we show in \cref{app:parameterization_scalar} that the overall dynamics only depends on the two parameters \[
\alpha_1 = \Delta t \gamma/(4M), \quad\alpha_2 = \Delta t^2/(\beta M).\]
as well as the schedule $z(t_k)$ and the  normalized velocities $\widetilde{v}_z(t_k) = \Delta t v_z(t_k)$. To fix these schedules, one can use the version given by \eqref{eq:schedule_function_z}-\eqref{eq:schedule_function_vz} with an appropriate function $f$ and choose the number of intermediate steps for a proposed jump as $K_T= \mathrm{ceil}\left(|\widetilde{Z}_{n+1} -  Z_n|/\widetilde{v}\right)$, with an effective normalized velocity $\widetilde{v}$. With this, the algorithm is then fully governed by the three parameters $\alpha_1, \alpha_2$, and $\widetilde{v}$.

Regarding the physical interpretation of $\alpha_1$ and $\alpha_2$, note first that the damping term in \eqref{eq:flucdissconstjarz1} and \eqref{eq:flucdissconstjarz2} is associated to the decorrelation of momenta over a time scale of $M/\gamma$. The parameter $\alpha_1$ therefore measures the inverse of this friction decorrelation time, rescaled by the time discretization. Furthermore, the equilibrium distribution of the velocities is a centered normal distribution with covariance $(\beta M)^{-1}$, which means that $\alpha_2$ is the squared displacement generated by a typical thermal velocity within one integration step.

For the parameter choice $\alpha_1=1$, the steered dynamics \eqref{eq:flucdissconstjarz1}-\eqref{eq:flucdissconstjarz2} reduce to overdamped Langevin dynamics (see \cref{sec:linearCV}). If $\xi$ is linear in addition, so that after a suitable change of basis the CV can be identified with the first $\ell$ degrees of freedom, the overall algorithm is equivalent to the `asymmetric' algorithm presented in our earlier work \cite{SCHONLE2025113806} (up to a slight change of the definition of the steered path: here we first move the CV and then the remaining degrees of freedom). Therein, we also provided the continuous-time limit (which is the one of infinite damping or zero mass) for this version of the algorithm together with the corresponding stochastic differential equation (SDE). For the choice of $\alpha_1=0$, the steered dynamics correspond to the Verlet scheme (Leapfrog dynamics) in the orthogonal degrees of freedom and are therefore a discretization of Hamiltonian dynamics with a time-dependent Hamiltonian. Note that in this case, the continuous-time limit of the dynamics $\Delta t \to 0$ simply corresponds to taking $\alpha_2\to 0$ and the acceptance of the move in \textbf{Step 3} is governed only by the (normalized) physical transition time $T/\sqrt{\beta M} = K_T\Delta t/\sqrt{\beta M}=\mathrm{ceil}\left(|\widetilde{Z}_{n+1} - Z_n|/\widetilde{v}\right)\sqrt{\alpha_2}$ (see \cref{fig:dumbell_determ_transition_time} in the appendix for a numerical illustration).

\section{Numerical examples}
\label{sec:numerics}
In this section, we apply the proposed CV-guided sampler to examples of increasing complexity, demonstrating in particular the advantage of deterministic underdamped steering dynamics compared to the previously considered overdamped version. We consider four model systems of increasing complexity: a theoretical model-system dubbed `Gaussian tunnel' (\cref{subsec:gauss_tunnel}), the one-dimensional $\phi^4$ model from statistical physics (\cref{subsec:phi4}), a dimer in a solvent (\cref{subsec:dimer}), and finally a 9-bead polymer in a solvent (\cref{subsec:polymer}). 
For each example, we performed a grid search within our chosen parameterization to optimize the performance of the algorithm. For the more realistic polymer example, we consider a 27-dimensional CV for which a generative model is necessary to build a proposal sampler adapted to the high-dimensional CV space. In our experiments, we demonstrate the efficacy of our approach using a normalizing flow to propose moves in CV space and converging the Monte Carlo simulation on the fully solvated system.

\subsection{Gaussian tunnel}
\label{subsec:gauss_tunnel}
We first consider a simple model with a linear CV as already introduced in previous work\cite{SCHONLE2025113806}. It is a toy model in $d=20$ dimensions with a designated one-dimensional variable $z$ serving as the CV and $d-1$ additional coordinates $(x^\perp_i)_{1\leq i\leq d-1}$. The distribution of~$z$ is given by a Gaussian mixture and the conditional distribution of the remaining coordinates is a Gaussian with a mean depending on $z$. More precisely,
\begin{align}
    \nu(\mathrm{d}z,\mathrm{d}x^\perp) &= \nu_{\rm CV}(\mathrm{d}z) \nu_\perp(\mathrm{d}x^\perp|z), \notag\\
    \nu_{\rm CV}(\mathrm{d}z) & = \frac{1}{\sqrt{2\pi}} \left(w \, \mathrm{e}^{-z^2/2} + (1-w)\, \mathrm{e}^{-(z-b)^2/2}\right)\mathrm{d}z,
    \label{eq:z_marginal}\\
    \nu_\perp(\mathrm{d}x^\perp|z) & = (2\pi )^{-(d-1)/2} \mathrm{det}(\Sigma)^{-1/2} \\
    &\times e^{\left(-{\frac {1}{2}}\left({x^\perp }-\mu(z)\right)^{\top}\Sigma ^{-1}\left({x^\perp }-\mu(z)\right)\right)} \, \mathrm{d}x^\perp.\notag
\end{align}
Here, $\mu(z)=\frac{b}{2} \cos(z \pi/b) (1,1,\dots,1)^\top$ is the conditional mean and $\Sigma$ a diagonal covariance matrix with ${{\Sigma }}_{ii}=\sigma_i^2$. The weight of the first mode is~$w=0.3$, while the distance between the modes is~$b=10$. The values of $\sigma_i$ were chosen from an evenly spaced grid in the interval $[0.5, 5]$.

To apply the algorithm from \cref{subsec:algorithm}, we use a Gaussian mixture model as the proposal sampler in CV space in \textbf{Step 1} which is the same as the target $\nu_\mathrm{CV}$ apart from a different mode weight $w=0.5$.
We identify $V(q) = V(z,x) = -\log \nu(z,x)$ with $\beta=1$, where, with a slight abuse of notation, we still denote by $\nu$ the density of the measure $\nu(\mathrm{d}q)$. We use the normalized version of the algorithm and schedule that only depends on $\alpha_1$, $\alpha_2$ and $\widetilde{v}$, detailed in \cref{subsec:Parameterization}. Since no Fixman term needs to be considered here, we simply use a linear schedule function $f(\tau)=\tau$.

To evaluate the performance, we consider (i) the acceptance rate for a fixed jump across the metastable basins from $z=0$ to $z=b$ normalized by the number of intermediate steps $K_T$ (as a measure of computational cost) and (ii) the inverse mode jump cost, which is the number of accumulated steps (calls to the force) that is needed on average to see a mode switch  between $z< b/2$ and $z>b/2$ when running the full algorithm.

Running a rough optimization through a grid search reveals that the optimum performance is obtained for $\alpha_1=0$, which corresponds to a deterministic steered dynamics (see \cref{fig:tunnel_optim} of \cref{app:subsec:gaussiantunnel}). We compare the deterministic ($\alpha_1=0$) and the overdamped ($\alpha_1=1$) steered dynamics in \cref{fig:Gaussians} using the two performance measures, which show consistent results. The deterministic dynamics offer a large performance increase over the overdamped dynamics  with a difference of around two orders of magnitude. In both cases, we observed that for too large values of $\alpha_2$ (a too coarse time discretization), the dynamics become unstable and the trajectories diverge, leading to zero acceptance of the overall move. The optimal value for $\alpha_2$ is then just below this threshold. Typical driven transition paths are shown for the optimum of the overdamped and the deterministic case in \cref{fig:GaussianTunnelTrajectories}. To verify the unbiasedness of the algorithm at the optimum, we show the marginals of $z$ and $x_0^\perp$ in \cref{fig:GaussianTunnelMarginals}. They indeed agree with the target measure and show that the algorithm corrects for the biased proposal distribution in CV space. 

Since the Fixman term and condition \cref{eq:zerovelocity_condition} on the schedule need not be considered in this simple case of a linear CV, we also show results for a modified version with a non-linear CV in \cref{subsec:AppNonlinGauss}. In this case, we indeed observe that a bias emerges when these ingredients are disregarded.
\begin{figure}
    \centering
    \includegraphics{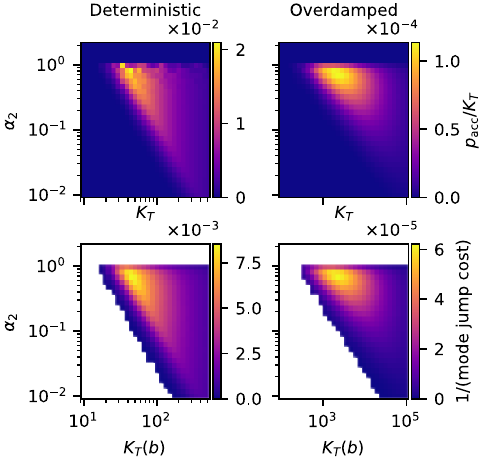}
    \caption{Performance of the algorithm for the Gaussian tunnel example in the special cases of deterministic ($\alpha_1 = 0$) and overdamped Langevin dynamics ($\alpha_1=1$). The top row shows the acceptance rate normalized by $K_T$ for a fixed jump from $z=0$ to $\widetilde{z}=b$ for different values of $K_T$, averaged over 10,000 random initializations of $x^\perp\sim\nu_\perp(x^\perp | z)$. The second row shows the inverse mode jump cost, averaged over 20 chains and 10,000 iterations, as a function of the number of steps $K_T(b)$ for a jump of distance $b$, which is related to the inverse of the dimensionless velocity $\widetilde{v}$ via $K_T(b) = \mathrm{ceil}(b/\widetilde{v})$. White space indicates that at least one chain never switched mode.}
    \label{fig:Gaussians}
\end{figure}

\begin{figure}
    \centering
    \includegraphics{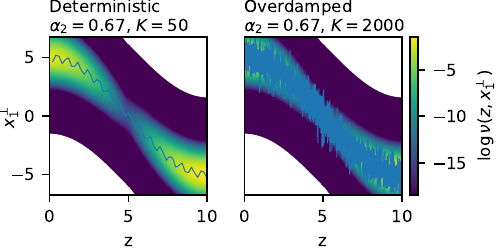}
    \caption{Typical driven transition path from one mode center to another, $z=0$ to $z=b=10$, for overdamped ($\alpha_1=1)$ and deterministic ($\alpha_1=0$) dynamics for the optimal $\alpha_2$ and $K_T=b/\widetilde{v}$.}
    \label{fig:GaussianTunnelTrajectories}
\end{figure}
\begin{figure}
    \centering
    \includegraphics{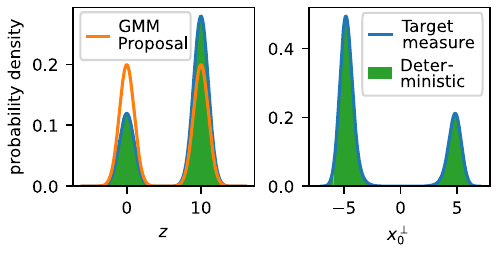}
    \caption{Gaussian Tunnel marginals for the optimum performance values $\alpha_1=0$, $\alpha_2 = 0.67$ and $\widetilde{v}=0.2$, corresponding to $K_T(b)=50$. Results are obtained by running 8 independent chains for 200,000 iterations.}
    \label{fig:GaussianTunnelMarginals}
\end{figure}
\subsection{\texorpdfstring{$\phi^4$}{Phi4} Model}
\label{subsec:phi4}
As a second example we consider a field model from statistical physics, the one-dimensional $\phi^4$ model (also considered previously\cite{SCHONLE2025113806}). It consists of $N$ degrees of freedom $\{\phi_i\}_{i=1}^N$ on a one-dimensional lattice and can be seen as a continuous version of the Ising model for ferromagnetism in solids. The target distribution is the Boltzmann--Gibbs probability distribution $\nu({\rm d}\phi)$ with the energy (using the simplified version of the model considered in Ref.~\onlinecite{gabrieAdaptiveMonteCarlo2022})
\begin{align}
    \beta V(\phi) = \frac{a\beta N}{2} \sum_{i=1}^{N+1}(\phi_i - \phi_{i-1})^2 + \frac{\beta}{4 a N}\sum_{i=1}^N (1-\phi_i^2)^2,
    \label{eq:phi4}
\end{align}
with Dirichlet boundary conditions $\phi_0=\phi_{N+1} = 0$. For sufficiently low temperatures~$1/\beta$, this system has two well-separated modes characterized by positive and negative  magnetization~$\overline{\phi}=\sum_{i=1}^N \phi_i / N$ concentrated at $\pm \overline{\phi}^*$ for some $\overline{\phi}^*$ which depends on $\beta$. We place ourselves in this bimodal phase by setting~$\beta=20$ and~$a=0.1$. In many practical settings, one will be interested in systems with an external field $h$, which corresponds to adding a term $h\sum_{i=1}^N \phi_i$ to the potential $V(\phi)$. For the purpose of testing the algorithm, however, the case $h=0$ is most illustrative since the two modes will have exactly the same weight due to the symmetry $V(\phi)=V(-\phi)$, and one can verify that the algorithm produces no bias with respect to the mode weights.

We use the magnetization $\overline{\phi}$ as the collective variable, which is a natural order parameter of the system. 
We run the algorithm on the system in a wavelet basis (see \cref{sec:wavelet_details} for details) since this order parameter then naturally emerges as the first component of the system. 
To fix the proposal sampler in CV space, we first ran a short simulation using the Metropolis Adjusted Langevin Algorithm (MALA)~\cite{RDF78,roberts1996}, which, due to the local proposal updates of MALA, always stays within one mode. From this data, the mean and standard deviation of the observed mode were then estimated as $\overline{\phi}^* = 0.79$ and $\sigma=0.06$. Consequently, we used as the proposal sampler a Gaussian mixture model with two modes centered at $\pm \overline{\phi}^*$ and of width $\sigma$.

Since we are again using a linear CV here, we use the same linear positions schedule $z(t_k)$ and constant velocity schedule $v_z(t_k)$ which we also used for the Gaussian tunnel. We compare the performance of the sampling algorithm over a grid of parameters $\alpha_1$, $\alpha_2$ and $\widetilde{v}$, using the mode jump cost as a criterion (the average number of calls to the force between transitions between $\overline{\phi}>0$ and $\overline{\phi}<0$). Optimal performance is again observed for the deterministic dynamics at $\alpha_2=0$ (see \cref{fig:Phi4optim} in \cref{subsec:AppPhi4}). Notably, as can be seen in \cref{fig:phi_comparison}, there is a performance increase by around two orders of magnitude passing from the overdamped to the deterministic dynamics. This is consistent with the experiments for the Gaussian tunnel. Away from the optimum, we see a white diagonal where no transitions were observed, corresponding to a particularly `bad' physical transition time $T/\sqrt{\beta M} = K_T\sqrt{\alpha_2}$. This might be a resonance effect due to the deterministic dynamics, but it is in any case far away from the parameter range of optimum performance. 

For the optimal parameters, we also run a modified experiment where we deliberately bias the CV proposal sampler by putting a larger weight on the mode corresponding to negative magnetization. As we show in \cref{fig:phi_marginals}, this bias is accurately corrected and the obtained samples align perfectly with the target density.

Driving the system from one mode to another corresponds to creating a domain of opposite orientation from either end of the lattice and gradually increasing its size. We indeed observe in \cref{fig:phi_transitionpaths} that `good' transitions with a high probability of acceptance follow one of these two transition paths. The existence of multiple transition paths does not pose a difficulty per se for the algorithm presented here: the acceptance of a proposal move is only determined by the work accumulated along the way and thus the intermediate bimodality is not a problem for accurate sampling of the two modes and their relative weights.
\begin{figure}
    \centering
    \includegraphics{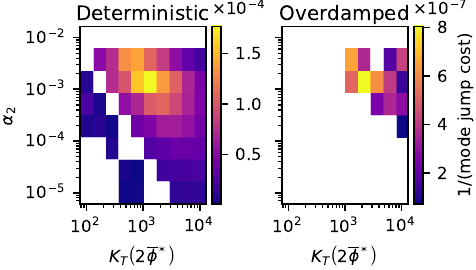}
    \caption{Performance of the algorithm applied to the $\phi^4$ model, estimated from 10 chains over 20,000 iterations. The x-axis represents the number of steps for a steered trajectory over the distance $2\overline{\phi}^*$ between the two modes. White space indicates that at least one chain never switched mode.}
    \label{fig:phi_comparison}
\end{figure}    
\begin{figure}
    \centering
    \includegraphics{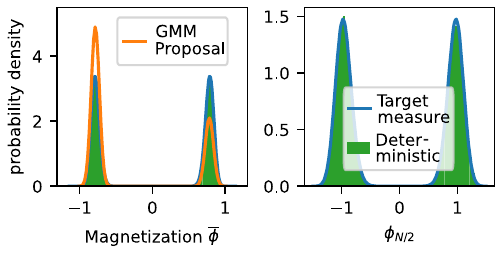}
    \caption{Marginal distribution of the $\phi^4$ model obtained from running 8 chains for $10^4$ steps with $\alpha_1=0$, $\alpha_2=0.0014$ and $\widetilde{v}=1.2\times 10^{-3}$ (the optimal parameters obtained from optimization in \cref{fig:phi_comparison}). In this case, a biased GMM proposal with weights $0.3$ and $0.7$ was used.}
    \label{fig:phi_marginals}
\end{figure}
\begin{figure}
    \centering
    \includegraphics{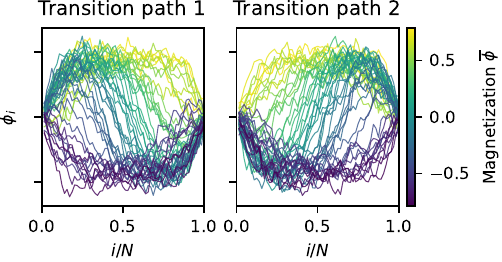}
    \caption{Two possible transition paths in $\phi^4$ model from positive to negative values of $\overline \phi$ with the same parameters as in \cref{fig:phi_marginals}}
    \label{fig:phi_transitionpaths}
\end{figure}
\subsection{Dimer in a solvent}
\label{subsec:dimer}
We now apply the algorithm on a model with a non-linear collective variable\cite{straub_molecular_1988,dellago_calculation_1999,lelievre-rousset-stoltz-book-10}. It consists of $N$ particles of equal masses $M$ in a two-dimensional box of side length $L$ with periodic boundary conditions. Two of the particles form a dimer; the remaining particles are designated as solvent particles. The solvent particles interact between them and with the dimer particles via the Weeks--Chandler--Andersen (WCA) pair potential, which is a truncated Lennard-Jones potential:
\begin{align*}
    V_\mathrm{WCA}(r) = \begin{cases}
        4\epsilon \left[\left(\frac{r^*}{r}\right)^{12} - \left(\frac{r^*}{r}\right)^6\right]+\epsilon\quad &\text{if } r\leq r_0,\\
        0 \quad &\text{if } r> r_0.
    \end{cases}
\end{align*}
Here, $r$ is the distance between two particles, $r_0=2^{1/6}r^*$ and $\epsilon$ and $r^*$ are positive constants. 
The two particles of the dimer interact via a double-well potential
\begin{align}
    V_\mathrm{D}(r) = h \left[1-\frac{(r-r_0-w)^2}{w^2}\right]^2,
\end{align}
with the two positive parameters $h$ and $w$. Denoting the particle positions by $\{q_i\}_{i=1}^N$ with each $q_i\in\mathbb{R}^2$, the total potential energy is thus given by
\begin{align}
    V(q) &= V_\mathrm{D}(|q_1-q_2|) + \sum_{3\leq i<j\leq N}V_\mathrm{WCA}(|q_i-q_j|) \notag \\
    &\quad + \sum_{i=1,2}\sum_{3\leq j \leq N} V_\mathrm{WCA}(|q_i-q_j|).
\end{align}
The minima of $V_\mathrm{D}$ are at $r=r_0$ and $r=r_0+2w$, giving rise to a compact state and stretched state with an energy barrier of height $h$. Two configurations for the compact and the stretched state are shown in \cref{fig:dimer_configs_example}.
\begin{figure}
    \centering
    \includegraphics{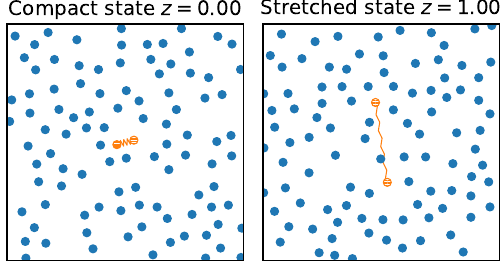}
    \caption{Typical configurations for the dimer in the compact and the stretched state.}
    \label{fig:dimer_configs_example}
\end{figure}

As a reaction coordinate, we use the normalized bond length
\begin{align}
    \xi(q) = \frac{|q_1-q_2| - r_0}{2w},
    \label{eq:dimer_reaction_coordinate}
\end{align}
so that the compact and the stretched states correspond to $\xi(q)=0$ and $\xi(q)=1$, respectively. As explained in \cref{subsec:dimer_fixman}, the Gram tensor is given by $G_M(q) = (2w^2 M)^{-1}$. Since it does not depend on $q$, the Fixman term \[
V_\mathrm{fix}(q) = - \frac{1}{2\beta} \log (2w^2 M)
\]
is therefore a constant that does not change the dynamics and need not be considered. The Lagrange multipliers for the constrained dynamics can be obtained analytically for this system and CV (see \cref{subsec:lagrane_multi_dimer}). We make use of this since it is numerically beneficial, but in principle, one could equally use numerical solvers as explained in \cref{subsubsec:Rattle}. 

We ran numerical experiments for parameters $h=w=2$, $r^*=\epsilon=1$ and $M=1$ with $N=100$ particles and a box of length $L=15$, in line with previous works\cite{lelievre-rousset-stoltz-book-10,lelievre_langevin_2012}. For the CV sampler we use a Gaussian mixture model with modes of equal weight at $z=0$ and $z=1$ and width $\sigma_{\mathrm{prop}} = 0.2$. We reject any unphysical proposals with $z<z_\mathrm{min} = -\frac{r_0}{2w}$ and $z>z_\mathrm{max} = (L/\sqrt{2} - r_0)/(2w)$ since the distance $|q_1-q_2|$ is always positive and the periodic boundary conditions (torus coordinates) impose a maximum possible distance. 

We compare the performance of the algorithm for different parameter values using a mode jump cost that counts the number of calls to the force required to go from a state with $z<0.1$ to a state with $z>0.9$ or vice versa, in line with the `mean residence duration' used in Ref.~\onlinecite{lelievre-rousset-stoltz-19}.

Once again, we observe a large advantage in using the deterministic algorithm with $\alpha_1=0$, whose performance for different values of $\alpha_2$ and $\widetilde{v}$ are shown in \cref{fig:dimer_optimalpha0}, while we see almost no mode switches for the overdamped algorithm within the same computational budget (see \cref{subsec:dimer_appendix}). 

In \cref{fig:dimer_free_energy}, we show the free energy profile obtained by running the algorithm for the optimized parameters until convergence together with the free energy in the absence of solvent particles (see \cref{appendix:dimer_free} for its derivation). As one can see, the presence of the solvent particles increases the likelihood of the compact state (lowering its free energy in comparison with the stretched state). The result agrees almost perfectly with a result obtained from thermodynamic integration (TI) (see Ref.~\onlinecite{lelievre_langevin_2012} for a theoretical introduction). In thermodynamic integration, the system is sampled under the constraint $\xi(q)=z_i$ for a sequential series of CV-values $z_i$ on a grid to estimate the mean force. This demonstrates that the inaccuracy of the proposal potential has been corrected by the algorithm. 

Note that this model was also considered in the original NCMC paper\cite{nilmeier2011nonequilibrium}. Contrary to our approach, however, they chose the full dimer configuration as the CV, $\xi(q)=(q_1,q_2)^\top$, and used a proposal kernel that samples a dimer configuration $\widetilde{Z}_{n+1}$ conditioned on the current configuration $Z_n$ with a new distance $|q_1-q_2|$ but the same overall rotation angle.
\begin{figure}
    \centering
    \includegraphics{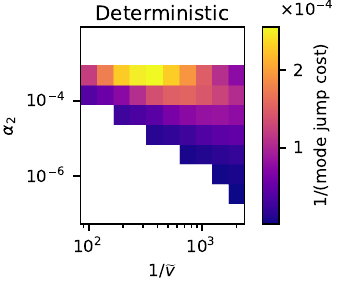}
    \caption{Performance comparison for the dimer model introduced in \cref{subsec:dimer} for the deterministic case $\alpha_1=0$, each experiment with 20 chains being run for 5,000 iterations. The optimal performance is observed at $\alpha_2 = 4.52 \times 10^{-4}$ and $1/\widetilde{v}=378$. White areas indicate that at least one chain never switched mode.}
    \label{fig:dimer_optimalpha0}
\end{figure}

\begin{figure}
    \centering
    \includegraphics{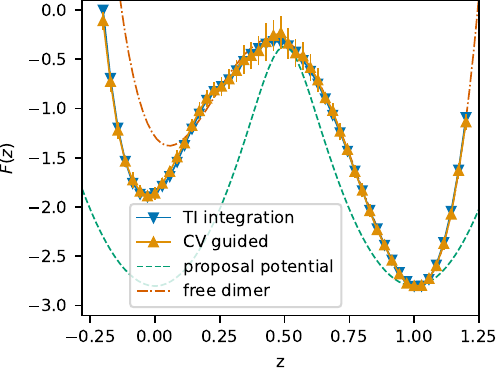}
    \caption{Free energy profile for the dimer model obtained from running the CV guided algorithm with the optimal parameters determined in \cref{fig:dimer_optimalpha0} with 20 chains and 50,000 iterations, shown alongside results from thermodynamic integration on a grid of 50 values and $10^6$ steps each. For comparison, we also show the free energy profile of the dimer without solvent particles, see \cref{appendix:dimer_free}, and the energy of the proposal distribution. An energy shift is applied to all of these curves such that they coincide at $z=1$. Error bars are estimated from the standard deviation over parallel MCMC chains of counts of the histogram bins.}
    \label{fig:dimer_free_energy}
\end{figure}

\subsection{Polymer in a solvent}
\label{subsec:polymer}
As a final example, we consider a toy molecular system that consists of a 9-bead polymer embedded in a three-dimensional box with 200 solvent particles, similar to a model introduced in Ref.~\onlinecite{tamagnoneCoarseGrainedMolecularDynamics2024}. The polymer beads and the solvent particles all interact via a Lennard-Jones (LJ) potential. Within the polymer, the nearest and second-nearest neighbor distances (and hence the binding angle) of the polymer beads are softly constrained with a harmonic potential. In addition, a bi-modal potential is imposed on the end-to-end distance of the polymer, giving rise to two stable confirmations in an open and a closed state, see \cref{fig:polymer_sketch}. Denoting by $\{q_i\}_{i=1,\dots,9}$ the polymer coordinates and by $\{q_i\}_{i=10,\dots,N}$ the solvent coordinates, the total energy is given by
\begin{align}
    V(q) &= \sum_{i=1}^8 V_1(|q_i - q_{i+1}|) + \sum_{i=1}^7 V_2(|q_i - q_{i+2}|) \notag\\
    &+ \sum_{\substack{1\leq i \leq 6\\ i+3\leq j \leq 9}} V_{\mathrm{LJ}}(|q_i - q_{j}|) +V_e(|q_1-q_9|) \notag \\
    &+ \sum_{\substack{1\leq i \leq 9\\ 10\leq j \leq N}} V_{\mathrm{LJ}}(|q_i - q_{j}|) \notag\\
    &+ \sum_{10\leq i<j \leq N} V_{\mathrm{LJ}}(|q_i - q_{j}|), 
\end{align}
with a Lennard-Jones interaction $V_{\mathrm{LJ}}$, harmonic potentials $V_\kappa(r) = \frac{k_\kappa}{2}(r-d_\kappa)^2$ for $\kappa\in \{1,2\}$ and a double-well potential (acting on the end-to-end distance) $V_e(r) = \frac{k_e}{4}  (\frac{r-d_e}{a})^4 - k_e (\frac{r-d_e}{a})^2$  with two local minima at $r = d_e \pm a\sqrt{2}$. Precise parameter values and implementation details are provided in \cref{sec:AppPolymer}. 
\begin{figure}
    \centering
    \includegraphics{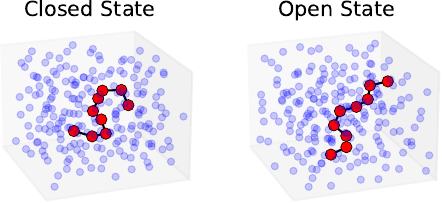}
    \caption{Open and closed state of the polymer (red) and the solvent particles (blue), characterized by the end-to-end distance}
    \label{fig:polymer_sketch}
\end{figure}

\begin{figure}
    \centering
    \includegraphics{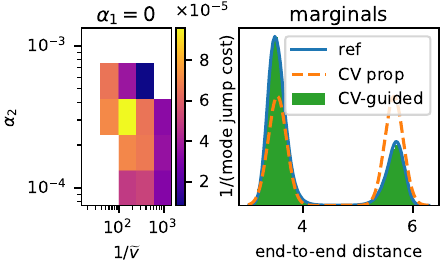}
    \caption{Simulation of the polymer with a one-dimensional CV (end-to-end distance). Left panel: Rough grid optimization of parameters for the deterministic case $\alpha_1=0$, each simulation run with 20 parallel chains and 2,000 iterations. Right panel: Marginals for a run of 20 chains and 100,000 iterations for the optimized parameters $1/\widetilde{v}=180$, $\alpha_2 = 2\times 10^{-4}$ (overall 8\% acceptance), shown alongside the CV proposal distribution and reference samples obtained from a biased simulation that has been reweighted.}
    \label{fig:polymer_1d}
\end{figure}
\begin{figure}
    \centering
    \includegraphics{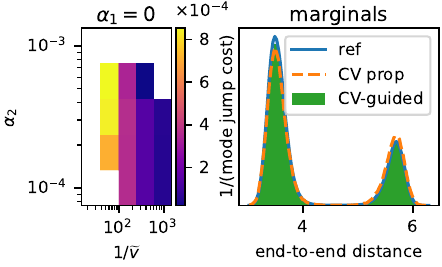}
    \caption{Simulation of the polymer with the entire polymer as the CV. Left panel: Rough grid optimization of parameters for the deterministic case $\alpha_1=0$, each simulation run with 20 parallel chains and 2,000 iterations. Right panel: Marginals for a run of 20 chains and 10,000 iterations for the optimized parameters $1/\widetilde{v}=60$, $\alpha_2 = 5\times 10^{-4}$ (99\% acceptance rate), shown alongside the marginalized CV proposal distribution and reference samples obtained from a biased simulation that has been reweighted.}
    \label{fig:polymer_27d}
\end{figure}

We test the algorithm for two different choices of CV: (i) the end-to-end distance, $\xi(q)=|q_1-q_9|$, a one-dimensional CV, and (ii) the entire polymer chain, $\xi(q) = (q_1,\ldots q_9)$, a 27-dimensional CV. For the latter, we define a proposal sampler with the help of a normalizing flow. 

When running simple MALA dynamics on this system, mode transitions are only observed very rarely, necessitating the use of an enhanced sampling technique. To obtain a reference dataset to compare against, we ran biased dynamics (removing the potential energy $V_e$ on the end-to-end distance) and unbiased them via reweighting.

An obvious choice of CV that resolves the bimodality is the end-to-end distance. Since it is one-dimensional, it is easy to construct a proposal distribution in the form of a Gaussian mixture. For simplicity, we extract the position of the modes and their width from short Langevin simulations on the full system (see \cref{sec:AppPolymer}). A rough grid search to optimize the parameters of the algorithm is shown in \cref{fig:polymer_1d}, with the best achieved inverse mode jump cost at around $9\times 10^{-5}$. In contrast, the overdamped algorithm with $\alpha_1=1$ does not show any transitions within the same parameter range and computational budget, giving an upper limit of the performance of at most $6\times 10^{-7}$; it is thus at least two orders of magnitude worse in performance.

Another possible choice of collective variable in this case is to use the positions of the beads of the entire polymer configuration. Here, a good proposal distribution requires a careful parameterization, for which we use a normalizing flow (see \cref{sec:AppPolymer}). Since our primary interest is in the performance of the algorithm, we trained our proposal sampler on already existing reference data. In practice, one could use adaptive schemes in which the CV-sampler is progressively trained while running the MCMC sampler\cite{gabrieAdaptiveMonteCarlo2022, tamagnoneCoarseGrainedMolecularDynamics2024}. Note that for this choice of CV, the steered dynamics in \textbf{Step 2} of the algorithm effectively only acts on the solvent particles, since the schedule already imposes the evolution of the entire polymer. Running the algorithm, we again observe that the deterministic dynamics ($\alpha_1=0$) gives good results with frequent mode switches while the overdamped version ($\alpha_1=1$) does not work at all. Strikingly, for the optimized parameters, we achieve an overall acceptance rate of 99\% with this algorithm. Since the proposals are drawn independently, this corresponds almost to the optimum performance with maximum decorrelation within the polymer degrees of freedom. A rough grid search for the deterministic algorithm is shown in \cref{fig:polymer_27d} together with the marginal distribution of the end-to-end distance. Comparing the optimal mode jump cost obtained with the two choices of CV, we see an advantage of around one order of magnitude in using the positions of all polymer beads over the end-to-end distance between first and last bead as a CV. This only accounts for the calls to the force, however, and not for the effort in learning the proposal sampler in CV space.

\section{Conclusion and outlook}
\label{sec:conclusion}
In this work, we present an MCMC algorithm with non-local proposals in a collective variable space. This CV space may be high-dimensional (and thus classic free energy adaptive biasing methods do not apply), provided a good proposal sampler can still be learned (e.g. by a normalizing flow for CVs of moderately large dimensions). We generalize related algorithms by considering the underdamped Langevin dynamics and a non-linear collective variable and we prove reversibility of the resulting scheme. We provide a detailed description of the algorithm so that it is readily implementable and we illustrate its performance on a number of examples, ranging from a toy low-dimensional example to molecular systems. We demonstrate a significant performance advantage in using Hamiltonian over overdamped steered Langevin dynamics. 
As mentioned in the introduction, our algorithm enters into frameworks like NCMC and HNMD, but we treat here all technical details arising from a non-linear CV and provide a numerical recipe for the use of general underdamped Langevin dynamics. 

There is a tradeoff in this kind of algorithm between the facility of finding a `good' CV and learning a good proposal sampler in that CV space. The system needs to be able to follow the constrained dynamics (\textbf{Step 2} of the algorithm) and still arrive at a low-energy configuration, for which it is vital that the CV resolves the metastability (see also Appendix B in Ref. \onlinecite{SCHONLE2025113806}). Evidently, this is easier to achieve allowing the CV to be higher-dimensional. In the (absurd) limiting case where the CV comprises all degrees of freedom of the system, the `steering' is trivially simple. At the same time, the higher-dimensional the CV space becomes, the harder it will be to parameterize and learn an efficient proposal sampler. Our polymer example with the 27-dimensional CV shows how there can be a sweet spot at intermediate dimensionality. Although we do not account for the training effort here, we demonstrate that a normalizing flow can be trained to such high accuracy that it can be used efficiently in the framework of this algorithm. When the proposal sampler has to be learned, one may resort to adaptive frameworks alternating between sampling and training, which have been demonstrated to be effective for normalizing flows\cite{gabrieAdaptiveMonteCarlo2022}, assuming that a good initial guess of the mode locations is available.

We leave certain questions for future investigations, notably the performance of the algorithm with respect to different choices of CV, as well as the choice of schedule (path) in CV space and the training of the proposal sampler. Furthermore, it would be very interesting to apply this algorithm to more complex molecular systems such as the conformational equilibrium of biomolecules, which is something we leave for future work.

\section*{Author declarations}
\subsection*{Conflict of interest}
The authors have no conflicts to disclose.

\section*{Data availability}
A running implementation of the algorithms and all code to reproduce results presented in this work are publicly available at \\
\url{https://github.com/cschoenle/mcmc_nonlincv}. 
Simulation  data can be reproduced at reasonable computational cost.

\begin{acknowledgments}
The authors thank warmly Alessandro Laio and Samuel Tamagnone for their help in setting up the last numerical example of the polymer. The authors remember fondly Samuel.
This work was supported by Hi! Paris, a French government grant managed by the Agence Nationale de la Recherche under the France 2030 program, reference ANR-23-IACL-0005, and supported as well by PR[AI]RIE-PSAI, also a French government grant managed by the Agence Nationale de la Recherche under the France 2030 program, reference ANR-23-IACL-0008. The works of T.L. and G.S. are partially
funded by the European Research Council (ERC) under the European Union’s Horizon 2020 research and innovation programme (project EMC2, grant agreement No 810367), and from the Agence Nationale de la Recherche through the grants ANR-19-CE40-0010-01 (QuAMProcs) and ANR-21-CE40-0006 (SINEQ).
This work was partially performed using HPC resources from GENCI–IDRIS (AD011015234R1).
\end{acknowledgments}

\bibliography{refs}
\clearpage
\appendix
\onecolumngrid
\section{Numerical implementation of the evolution of momenta}
\label{app:numericalformulation}
For practical purposes, one can reformulate \eqref{eq:flucdissconstjarz1} (and \eqref{eq:flucdissconstjarz2} in the same fashion) in the following equivalent form that is more suitable to numerical implementation, see Sec.~5.4.2 in Ref.~\onlinecite{lelievre_langevin_2012} for further explanations:
\begin{equation}
\left \{ 
\begin{aligned}
& p^{k+1/4} = p^k  - \frac{\dt}{4} \gamma M^{-1}  
\pare{ p^k + p^{k+1/4} - 2\nabla\xi G_{M}^{-1}(q^k) 
v_z(t_{k}) } \\
& \phantom{p^{k+1/4} =} + \sqrt{\frac{\dt}{2}} \, \sigma \, {\mathcal G}^k 
+ \nabla\xi(q^k)\, \lambda^{k+1/4},\\
& \nabla\xi(q^k)^\top M^{-1}p^{k+1/4} = v_z(t_{k}),  \qquad\qquad(C_p)
\end{aligned} \right.
\end{equation}
where the Lagrange multiplier $\lambda^{k+1/4}$ is associated with the constraint $(C_p)$. It can be determined analytically by solving
\begin{align*}
v_z(t_{k}) = \nabla\xi(q^k)^\top &M^{-1} (\mathrm{Id}  + \frac{\dt}{4} \gamma M^{-1})^{-1} \\
&\pare{ (\mathrm{Id}  - \frac{\dt}{4} \gamma M^{-1})p^k  
 + \frac{\dt}{2} \gamma M^{-1}\nabla\xi G_{M}^{-1}(q^k) 
v_z(t_{k}) + \sqrt{\frac{\dt}{2}} \, \sigma \, {\mathcal G}^k 
+ \nabla\xi(q^k) \, \lambda^{k+1/4}},
\end{align*}
assuming that $\nabla\xi(q^k)^\top M^{-1} (\mathrm{Id}  + \frac{\dt}{4} \gamma M^{-1})^{-1} \nabla\xi(q^k)$ is invertible.
\section{Proof of the reversibility of the algorithm}
\label{sec:appendix_reversibility}
The objective of this section is to prove the reversibility of the algorithm introduced in Section~\ref{subsec:algorithm} (namely~\cref{theo:reversbility}). This is done in three steps. In Section~\ref{sec:appendix_reversibility1}, we introduce the necessary material, in particular concerning measures on submanifolds associated with the collective variable $\xi$. Then, in Section~\ref{sec:appendix_reversibility2}, we present fundamental properties of the steered dynamics~\eqref{eq:flucdissconstjarz1}-\eqref{eq:flucdissconstjarz2}, and in particular the Jazynski-Crooks equality. This equality is the cornerstone of the proof of reversibility, which is finally provided in Section~\ref{sec:appendix_reversibility3}.
\subsection{Introducing notation}\label{sec:appendix_reversibility1}
This subsection largely uses tools and results from Ref.~\onlinecite{lelievre_langevin_2012}. Let us introduce a number of quantities which we will need later on. 
A central quantity of interest is the free energy $F:\mathbb{R}^{\ell}\rightarrow \mathbb{R}$ associated with the target measure $\nu$ and the collective variable $\xi$, defined as (see \eqref{eq:Sigmaz} for the definition of $\Sigma(z)$)
\begin{align}
    F(z) = -\frac{1}{\beta} \log \int_{\Sigma(z)} \mathcal{Z}_q^{-1}e^{-\beta V(q)}\delta_{\xi(q)-z}(\mathrm{d}q).
    \label{eq:free_energy}
\end{align}
The phase space associated with constraints on the collective variable $\xi$ and effective velocity $v_\xi$ is denoted as
\begin{align}
    \Sigma_{\xi, v_\xi}(z, v_z)  = \left\{(q,p)\in\mathbb{R}^{2d}\, \Big | \, \xi(q) = z, \, v_\xi(q,p)=v_z \right\}.
\end{align}
For a fixed position $q$, we denote the affine space of constrained momenta as
\begin{align}
    \Sigma_{v_\xi(q,\cdot)}(v_z) = \left\{p\in \mathbb{R}^{d}\,\middle| \, v_\xi(q,p) = v_z\right\}.
\end{align}
Note that $T^*_q\Sigma(z) = \Sigma_{v_\xi(q,\cdot)}(0)$, which was defined in \eqref{eq:cotangent_zero_speed}.
The Boltzmann--Gibbs probability distribution can also be defined on these constrained spaces as follows:
\begin{align}
    \begin{cases}
    \ds \mu_{\Sigma_{\xi,v_\xi}(z,v_z)}(\mathrm{d}q\, \mathrm{d}p ) := \frac{ {\rm e}^{-\beta H(q,p)}}{\mathcal{Z}_{z,v_z}} \sigma_{\Sigma_{\xi,v_\xi}(z,v_z)}(\mathrm{d}q\, \mathrm{d}p ), \\
    \ds \mathcal{Z}_{z,v_z} := \int_{\Sigma_{\xi,v_\xi}(z,v_z)} {\rm e}^{-\beta H} \, \mathrm{d} \sigma_{\Sigma_{\xi,v_\xi}(z,v_z)} ,
  \end{cases}
  \label{eq:constrainedBoltzmann}
\end{align}
where $\sigma_{\Sigma_{\xi,v_\xi}(z,v_z)}(\mathrm{d}q\, \mathrm{d}p)$ is the phase space Liouville measure on $\Sigma_{\xi,v_\xi}(z,v_z)$.

Using the Fixman term $V_\mathrm{fix}$ introduced in \cref{eq:deffixman}, it is possible to relate the free energy~\eqref{eq:free_energy} to the normalisation of the restricted Boltzmann--Gibbs measure \eqref{eq:constrainedBoltzmann}. More precisely, the modified potential  $\widetilde{V}(q) = V(q) + V_\mathrm{fix}(q)$ gives rise to a new Hamiltonian
denoted $\widetilde{H}(p,q) = H(q,p)+V_\mathrm{fix}(q)$ and an associated constrained measure $\widetilde{\mu}_{\Sigma_{\xi,v_\xi}}$  with normalizing constant $\widetilde{\mathcal{Z}}_{z,v_z}$ (using the formulas in \eqref{eq:constrainedBoltzmann} with $H$ replaced by $\widetilde{H}$).
We then obtain the following reformulation of free energy differences in terms of ratios of the normalizing constants $\widetilde{\mathcal{Z}}_{z,v_z}$.
\begin{lem}
For all $(z,v_z) \in \R^\ell \times \R^\ell$, it holds that
\begin{align*}
    \widetilde{\mathcal{Z}}_{z,v_z} = (\det M)^{\frac 1 2} \left(2\pi \beta^{-1}\right)^{\frac{d-\ell}{2}} \int_{\Sigma(z)} e^{-\beta V(q)} \exp\left(-\frac{\beta}{2}v_z^\top G_M(q)^{-1} v_z\right) \delta_{\xi(q)-z} (\mathrm{d}q).
\end{align*}    
In particular, for any $z_1,z_2 \in \R^\ell$,
\begin{align}
    \frac{\widetilde{\mathcal{Z}}_{z_1,0}}{\widetilde{\mathcal{Z}}_{z_2,0}} = e^{-\beta (F(z_1) - F(z_2))}.
    \label{eq:free_energy_relations}
\end{align}
\label{lem:free_energy_relations}
\end{lem}
\begin{proof}
    The proof relies on Ref.~\onlinecite{lelievre_langevin_2012}.
    From Sec.~2.3 therein, we restate the following two identities: By the co-area formula, one can show that 
    \begin{equation}
    \label{eq:coarea}
    \delta_{\xi(q)-z}(\mathrm{d}q) = 
    \pare{{\rm det  }\, M}^{-1/2} \big| {\rm det  } \, G_M(q) \big|^{-1/2} \,
    \sigma^{M}_{\manq(z)}(\mathrm{d}q ).
  \end{equation}
  Furthermore, the measure $\sigma_{\Sigma_{\xi,v_\xi}(z,v_z)}$ can be written as a product of two measures:
\begin{equation}
  \label{eq:surfacemeas}
  \sigma_{\manq_{\xi,v_\xi}(z,v_z)}(\mathrm{d}q \, \mathrm{d}p) =
  \sigma^{M^{-1}}_{\manq_{v_\xi(q,\cdot)}(v_z)}(\mathrm{d}p) \, \sigma^{M}_{\manq(z)}(\mathrm{d}q),
 \end{equation}
 where we denote
by $\sigma^{M}_{\manq(z)}(\mathrm{d}q)$ the surface measure on $\manq(z)$ induced by the scalar
product $\langle q,\tilde{q}\rangle_M = q^\top M \tilde{q}$ on $\R^{d}$, and, 
for a given $q \in \manq(z)$, by
$\sigma^{M^{-1}}_{\manq_{v_\xi(q,\cdot)}(v_z)}(\mathrm{d}p)$ the
surface measure on the affine
space $\manq_{v_\xi(q,\cdot)}(v_z)$, 
induced by the scalar product
$\langle p,\tilde{p}\rangle_{M^{-1}} = p^\top M^{-1} \tilde{p}$ on~$\R^{d}$.
For more precise definitions of these measures, we refer 
to Ref.~\onlinecite[Sections~3.2.1 and~3.3.2]{lelievre-rousset-stoltz-book-10}
and the references therein.

Using the above identities, we write
\begin{align*}
    \widetilde{\mathcal{Z}}_{z,v_z} &= \int_{\Sigma_{\xi, v_\xi}(z, v_z)} e^{-\beta \widetilde{H}(q,p)} \sigma_{\Sigma_{\xi, v_\xi(z, v_z)}} (\mathrm{d} q \, \mathrm{d}p) \\
    &= \int_{\Sigma_{\xi, v_\xi}(z, v_z)} e^{-\beta H(q,p)} (\det(G_M(q)))^{-1/2} \sigma^{M^{-1}}_{\Sigma_{v_\xi(q,\cdot)}(v_z)}(\mathrm{d} p) \, \sigma^M_{\Sigma(z)} (\mathrm{d} q)\\
    &= \int_{\Sigma(z)} \int_{\Sigma_{v_\xi(q,\cdot)}(v_z)} \exp\left(-\frac{\beta}{2}p^\top M^{-1}p\right)\sigma^{M^{-1}}_{\Sigma_{v_\xi(q,\cdot)}(v_z)}(\mathrm{d} p) \; e^{-\beta V(q)} (\det(G_M(q)))^{-1/2} \sigma^M_{\Sigma(z)} (\mathrm{d} q)  \\
    &= \int_{\Sigma(z)} \left(2\pi \beta^{-1}\right)^{\frac{d-\ell}{2}} \exp\left(-\frac{\beta}{2}v_z^\top G_M(q)^{-1} v_z\right) \; e^{-\beta V(q)} (\det(G_M(q)))^{-1/2} \sigma^M_{\Sigma(z)} (\mathrm{d} q)  \\    
    &= (\det M)^{\frac 1 2} \left(2\pi \beta^{-1}\right)^{\frac{d-\ell}{2}} \int_{\Sigma(z)} e^{-\beta V(q)} \exp\left(-\frac{\beta}{2}v_z^\top G_M(q)^{-1} v_z\right) \delta_{\xi(q)-z} (\mathrm{d}q).
\end{align*}
Equation~\eqref{eq:free_energy_relations} follows by taking $v_z=0$. 
For the last but one equality in the series of equalities above, we use a computation from Sec.~5.1 in Ref.~\onlinecite{lelievre_langevin_2012}, namely, for fixed $q$,
\begin{align*}
    &\int_{\Sigma_{v_\xi(q,\cdot)}(v_z)} \exp\left(-\frac{\beta}{2}p^\top M^{-1}p\right)\sigma^{M^{-1}}_{\Sigma_{v_\xi(q,\cdot)}(v_z)}(\mathrm{d} p) \\
    &= \exp\left(-\frac{\beta}{2}v_z^\top G_M(q)^{-1} v_z\right) \int_{\Sigma_{v_\xi(q,\cdot)}(0)} \exp\left(-\frac{\beta}{2}p^\top M^{-1}p\right) \sigma^{M^{-1}}_{\Sigma_{v_\xi(q,\cdot)}(0)} (\mathrm{d}p) \\
    &= \left(2\pi \beta^{-1}\right)^{\frac{d-\ell}{2}} \exp\left(-\frac{\beta}{2}v_z^\top G_M(q)^{-1} v_z\right).
\end{align*}
The first equality follows from a change of variable: the space $\Sigma_{v_\xi(q,\cdot)}(v_z)$ is composed of momenta fulfilling $0 = v_\xi(q,p) - v_z = \nabla\xi(q)^\top M^{-1} \left(p - \nabla\xi(q)G_M^{-1}(q)v_z\right)$, and therefore, when replacing $p$ by $p - \nabla\xi(q)G_M^{-1}(q)v_z$, the resulting integral is taken over the space $\Sigma_{v_\xi(q,\cdot)}(0)$.
\end{proof}

\begin{rem}
\label{rem:free_energy_relations_nofixman}
    When the Gram tensor $G_M(q)$ is independent of $q$, the introduction of the Fixman term becomes unnecessary and one furthermore has that
    \begin{align}
    \frac{\widetilde{\mathcal{Z}}_{z_1,v_z}}{\widetilde{\mathcal{Z}}_{z_2,v_z}} = \frac{\mathcal{Z}_{z_1,v_z}}{\mathcal{Z}_{z_2,v_z}} = e^{-\beta (F(z_1) - F(z_2))}
    \label{eq:free_energyrelation_nofixman}
\end{align}
for any velocity $v_z\in \mathbb{R}^\ell$.
\end{rem}

\subsection{Properties of the steered dynamics}\label{sec:appendix_reversibility2}

In this section, we state a number of properties of the discretization scheme introduced in \eqref{eq:flucdissconstjarz1}-\eqref{eq:flucdissconstjarz2}. Since we follow dynamics for the energy modified by the Fixman term, in principle we would need to put a tilde on all associated measures and quantities, but we omit it here for the sake of readability and just consider the dynamics for a general potential $V$.
\subsubsection{Forward steered dynamics}
\label{subsubsec:appendix_forw}
Let us make two remarks concerning the forward steered dynamics \eqref{eq:flucdissconstjarz1}-\eqref{eq:flucdissconstjarz2}
\begin{rem}
\label{rem:existence_lagrange_multip}
    In \cref{eq:Verletconstswitched}, the constraint $(C_p)$  yields a unique well-defined Lagrange multiplier $\lambda^{k+3/4}$ which can be obtained by solving the linear equation
    \begin{align*}
        v_z(t_{k+1}) &= \nabla \xi (q^{k+1})^{T} M^{-1}  p^{k+1/2} - \frac{\dt}{2} \nabla \xi (q^{k+1})^{T} M^{-1}  \nabla V (q^{k+1}) + \nabla \xi (q^{k+1})^{T} M^{-1} \nabla \xi(q^{k+1}) \lambda^{k+3/4},
    \end{align*}
    since the Gram tensor $G_M(q^{k+1}) = \nabla \xi (q^{k+1})^{T} M^{-1} \nabla \xi(q^{k+1})$ is invertible by assumption.
    In contrast, the constraint $(C_q)$ leads to a non-linear equation for $\lambda^{k+1/2}$, which is solved using Newton's method. 
    Regarding uniqueness of the solution and  convergence of this solver, we require a reversibility property between the forward dynamics and the backward steered dynamics introduced in \cref{subsubsec:bckwd_dynamics}. This property is explained at the end of \cref{subsubsec:symplecticityflow}, defined by \eqref{eq:rev_forwbackw_flows}. Loosely speaking, one needs to have a consistent procedure which gives the same solution when following the forward or backward dynamics.
\end{rem}
\begin{rem}
\label{rem:velocity_schedule_choice}
    Unlike the standard RATTLE discretization, our formulation considers time-evolving constraints. Moreover, in contrast to Ref.~\onlinecite[Sec. 5.4.1]{lelievre_langevin_2012}, where such versions of RATTLE are also considered, it uses a separate schedule for the velocity $(v_z^{(Z_n, \widetilde{Z}_{n+1})}(t_k))_{0\le n \le  K_{T}}$, that can in principle be chosen independently of the position schedule $(z^{(Z_n, \widetilde{Z}_{n+1})}(t_k))_{0\le k \le  K_{T}}$. Still, it is numerically advantageous to use a schedule with a meaningful limit in continuous time, such that (in the limit $\Delta t \to 0$) the position schedule converges to $\xi(q(t)) = z(t)$ where $z$ is a fixed continuously-differentiable function  and such that the velocity schedule converges to $v_z(t) = \dot{z}(t)$. To see this, we write the unconstrained part of the update in \eqref{eq:Verletconstswitched} as
    \[
        \widetilde{q}^{k+1} = q^{k} + \dt \ M^{-1} p^{k+1/4} - \frac{\dt^2}{2} \ M^{-1} \nabla V (q^{k}),
    \] which is followed by a projection with $q^{k+1} = \widetilde{q}^{k+1} + \Delta t M^{-1}\nabla\xi(q^k)\lambda^{k+1/2}$, where $\lambda^{k+1/2}$ is such that $\xi(q^{k+1}) = z(t_{k+1})$. Two examples of velocity schedule that converge to $v_z(t) = \dot{z}(t)$ in the continuous-time limit are (1) trivially taking $v_z(t_k) = \dot{z}(t_k)$ (as in \eqref{eq:schedule_function_z}-\eqref{eq:schedule_function_vz} in \cref{subsubsec:schedule}) or (2) the discretization \[v_z(t_k) = \frac{z(t_{k+1}) - z(t_{k-1})}{2\Delta t},\,  0 \leq k \leq K_T.\] 
    For this second choice, satisfying the condition \eqref{eq:zerovelocity_condition} ($v_z(t_0) = v_z(t_{K_T}) = 0$) implies extending the position schedule to times $t_{-1}$ and $t_{K_T+1}$ and setting $z(t_{-1}) = z(t_1)$ and $z(t_{K_T+1}) = z(t_{K_T-1})$.    
    For both choices, one has that $\xi(\widetilde{q}^{k+1}) = z(t_{k+1}) + \mathcal{O}(\Delta t^2)$ (using that $K_T=T/\Delta t = \mathcal{O}(\Delta t^{-1})$, where $T$ is the physical transition time) and thus the requirement $(C_q)$ is already approximately fulfilled by $\widetilde{q}^{k+1}$, facilitating convergence of the Newton solver for $\lambda^{k+1/2}$. Note, however, that the simplest choice of discretization $v_z(t_k) = \frac{z(t_{k+1}) - z(t_{k})}{\Delta t}$ which was used in Ref.~\onlinecite[Sec. 5.4.1]{lelievre_langevin_2012}, is not possible in our context since it would violate the reversibility property \eqref{eq:rev_velocities}. 
\end{rem}
\subsubsection{Backward steered dynamics}
\label{subsubsec:bckwd_dynamics}
For given position schedule $(z(t_k))_{0\le k \le  K_{T}}$ and velocity schedule $(v_z(t_k))_{0\le k \le  K_{T}}$, we define the following backward process. Starting from $(q^{\mathrm{b},0}, p^{\mathrm{b},0})\in \Sigma_{\xi,v_\xi}(z(t_{K_T}),v_z(t_{K_T}))$, iterate the following procedure on $0 \leq k' \leq K_T-1$:
\begin{align}
&  
\dps    p^{{\rm b},k'+1/4} =   p^{{\rm b},k'} -\frac{\dt}{4} \gamma_P(q^{{\rm b},k'}) M^{-1}(p^{{\rm b},k'+1/4}+p^{{\rm b},k'}) + \sqrt{ \frac{\dt}{2}} \sigma_P(q^{{\rm b},k'}) {\mathcal G}^{{\rm b},k'}\label{eq:flucdissconstjarzbck1} \\[6pt]
&\begin{cases}
 \dps  p^{{\rm b},k'+1/2} = p^{{\rm b},k'+1/4} + \displaystyle{\frac{\dt}{2} \nabla V (q^{{\rm b},k'})} + \nabla \xi(q^{{\rm b},k'}) \lambda^{{\rm b},k'+1/2} &\\[6pt]
 \dps  q^{{\rm b},k'+1} = q^{{\rm b},k'} - \dt \ M^{-1} p^{{\rm b},k'+1/2}  \\[6pt]
 \dps   \xi(q^{{\rm b},k'+1})  = z(t_{K_T-k'-1}) & (C_q) \\[6pt]
 \dps  p^{{\rm b},k'+3/4} = p^{{\rm b},k'+1/2} + \displaystyle{\frac{\dt}{2} \nabla V (q^{{\rm b},k'+1})} +\nabla \xi(q^{{\rm b},k'+1}) \lambda^{{\rm b},k'+3/4}&\\[6pt]
  \dps   \nabla \xi (q^{{\rm b},k'+1})^{T} M^{-1} p^{{\rm b},k'+3/4} = v_z(t_{K_T-k'-1}) &(C_p)
\end{cases}\label{eq:Verletconstswitchedbck}\\
&   \dps    p^{{\rm b},k'+1} = p^{{\rm b},k'+3/4} -\frac{\dt}{4} \gamma_P(q^{{\rm b},k'+1}) M^{-1}(p^{{\rm b},k'+3/4}+p^{{\rm b},k'+1})+ \sqrt{\frac{\dt}{2}} \sigma_P(q^{{\rm b},k'+1}) {\mathcal G}^{{\rm b},k'+1/2}, \label{eq:flucdissconstjarzbck2}
\end{align}
where $({\mathcal G}^{{\rm b},k'})$ and $({\mathcal G}^{{\rm b},k'+1/2})$ are independent sequences of i.i.d. centered Gaussian random vectors with covariance matrix~$\mathrm{Id}_{d}$ and $\lambda^{{\rm b},k'+1/2}$ and $\lambda^{{\rm b},k'+3/4}$ are Lagrange multipliers such that the conditions $(C_q)$ and $(C_p)$ are fulfilled. 
\subsubsection{Relation between processes with opposite endpoints}
Thanks to the reversibility property of the schedule \eqref{eq:rev_positions}-\eqref{eq:rev_velocities}, we can relate the law of the forward and backward steered dynamics with schedules with swapped endpoints one to another: With the right initialization, the solution of the backward dynamics for a schedule with endpoints $(z,\widetilde{z})$ has the same law as the forward dynamics for a schedule with endpoints $(\widetilde{z},z)$ up to a momentum flip. We make the statement more precise in the following lemma.
\begin{lem}
Consider a forward process~$\left((q^k, p^k, p^{k+1/4}, p^{k+3/4})_{0\leq k\leq K_T-1}, q^{K_T}, p^{K_T}\right)$ for a schedule $(z^{( \widetilde{z},z)}(t_k))_{0\le k \le  K_{T}}$, $(v_z^{(\widetilde{z},z)}(t_k))_{0\le k \le  K_{T}}$ with initial conditions 
\begin{equation}
    (q^0,p^0) \sim \mu_{\manq_{\xi,v_\xi} \pare{ z^{(\widetilde{z},z)}(t_0), v_z^{(\widetilde{z},z)}(t_0)} }(\mathrm{d}q\, \mathrm{d}p ),
\end{equation}
and a backward process~$\left((q^{\mathrm{b},k}, p^{\mathrm{b},k}, p^{\mathrm{b},k+1/4}, p^{\mathrm{b},k+3/4})_{0\leq k \leq K_T-1}, q^{\mathrm{b}, K_T}, p^{\mathrm{b}, K_T}\right)$ for a schedule with swapped endpoints $(z^{(z,\widetilde{z})}(t_k))_{0\le k \le  K_{T}}$, $(v_z^{(z,\widetilde{z})}(t_k))_{0\le k \le  K_{T}}$ and initial conditions 
\begin{equation}
  (q^{ { \rm b},0},p^{ {\rm b},0}) \sim \mu_{\manq_{\xi,v_\xi} \pare{ z^{(z, \widetilde{z})}(t_{K_T}), v_z^{(z,\widetilde{z})}(t_{K_T})  } }(\mathrm{d}q\, \mathrm{d}p ).
\end{equation}
Then, this backward process has the same law as the momentum flipped version of the forward process, $\left((q^{k},-p^k, -p^{k+1/4}, -p^{k+3/4})_{0\leq k \leq K_T-1}, q^{K_T}, -p^{K_T}\right)$.
    \label{lem:forw_backw_same_law}
\end{lem}
\begin{proof}
    The initial distribution of the forward process satisfies  
    $\mu_{\Sigma_{\xi,v_\xi}\left(z^{{( \widetilde z, z)}}(t_{0}), v_z^{{( \widetilde z, z)}}(t_{0})\right)} = \mu_{\Sigma_{\xi,v_\xi}\left(z^{(z,  \widetilde z)}(t_{K_T}), -v_z^{(z, \widetilde z)}(t_{K_T})\right)}$ by the symmetry properties \eqref{eq:rev_positions}-\eqref{eq:rev_velocities} of the schedule. Since \[(q,p) \in \Sigma_{\xi,v_\xi}\left(z, -v_z\right) \Leftrightarrow (q,-p) \in \Sigma_{\xi,v_\xi}\left(z, v_z\right),\]
    and from the definition of the constrained measure in \eqref{eq:constrainedBoltzmann} with the symmetry of the Hamiltonian $H(q,p) = H(q,-p)$, this implies that $(q^{\mathrm{b},0}, p^{\mathrm{b},0})$ and $(q^{0}, -p^{0})$ are indeed drawn from the same distribution. From there, using again the symmetry of the schedule, one can easily verify that the equations describing the forward dynamics with endpoints $(\widetilde z,z)$ for $(q^k,-p^k )$ are exactly the same as the equations describing the backward dynamics with end points~$(z, \widetilde z)$ for $(q^{\mathrm{b},k}, p^{\mathrm{b},k})$.
\end{proof}
\subsubsection{Symplecticity and reversibility of the deterministic flows}
\label{subsubsec:symplecticityflow}
\begin{lem}
\label{lem:flow_symplectic}
Consider the numerical flow 
    \begin{equation}
\Phi^k \ : \ 
\left\{ \begin{array}{ccc}
  \dps \manq_{\xi,v_\xi}\left(z(t_k), v_z(t_k)\right) 
  & \to & \dps \manq_{\xi,v_\xi}\left(z(t_{k+1}), v_z(t_{k+1})\right) \\[10pt]
  (q^k,p^{k+1/4}) & \mapsto & (q^{k+1},p^{k+3/4}) 
\end{array}\right.
\label{eq:fwd_flow}
\end{equation} associated with the deterministic part of the forward dynamics \eqref{eq:Verletconstswitched} as well as the flow
\begin{equation}
\Phi^{{\rm b},k'} \ : \ 
\left\{ \begin{array}{ccc}
  \dps \manq_{\xi,v_\xi}\left(z(t_{K_T-k'}), v_z(t_{K_T-k'})\right) 
  & \to & \dps \manq_{\xi,v_\xi}\left(z(t_{K_T-k'-1}), v_z(t_{K_T-k'-1})\right) \\[10pt]
  (q^{{\rm b},k'}, p^{{\rm b},k'+1/4}) & \mapsto & (q^{{\rm b},k'+1}, p^{{\rm b},k'+3/4}) 
\end{array}\right.
\label{eq:bckw_flow}
\end{equation}
associated with the deterministic part of the  backward dynamics \eqref{eq:Verletconstswitchedbck}. 
    
The flow $\Phi^k$ is a symplectic map for all $k\in \{0,\ldots,K_T-1\}$, and the same holds true for $\Phi^{{\rm b},k'}$ for all $k'\in \{0,\ldots,K_T-1\}$.
\end{lem}
To prove this statement, we show that the 2-form $\mathrm{d}q \wedge \mathrm{d} p$ is preserved, where $\wedge$ denotes the wedge-product (see Ref.~\onlinecite{McDuffSalamonSymplecticTop} for a textbook introduction to symplectic structures). For the simplest version with a constant constraint, this is a standard result found for example in Refs.~\onlinecite{hairer-lubich-wanner-06} and \onlinecite{leimkuhler_symplectic_1994}. 
    We follow here the same arguments as in Ref.~\onlinecite{leimkuhler_symplectic_1994}. The differentials obey:
\begin{align}
    \mathrm{d}  p^{k+1/2} &= \mathrm{d} p^{k+1/4} - \frac{\dt}{2} \mathrm{d} (\nabla V (q^{k})) + \mathrm{d} (\nabla \xi(q^k) \lambda^{k+1/2}) &\\
    \mathrm{d} q^{k+1} &= \mathrm{d} q^{k} + \dt \ M^{-1} \mathrm{d} p^{k+1/2}  \label{eq:diff_qkplus1}&\\
    \nabla\xi(q^{k+1})^\top \mathrm{d} q^{k+1}  &= 0& \label{eq:cq_diff}\\
    \mathrm{d} p^{k+3/4} &= \mathrm{d} p^{k+1/2} - \frac{\dt}{2} \mathrm{d} (\nabla V (q^{k+1})) +\mathrm{d} (\nabla \xi(q^{k+1}) \lambda^{k+3/4}).& 
\end{align}
Denoting the Hessian of $V$ by $\nabla^2 V$, we have $\mathrm{d} (\nabla V (q^{k})) = \nabla^2 V(q^k) \mathrm{d} q^k$. Taking the wedge product at the endpoints, we have
\begin{align}
    \mathrm{d} q^{k+1} \wedge \mathrm{d} p^{k+3/4} &= \mathrm{d} q^{k+1} \wedge \left( \mathrm{d} p^{k+1/2} - \frac{\dt}{2} \mathrm{d} (\nabla V (q^{k+1})) +\mathrm{d} (\nabla \xi(q^{k+1}) \lambda^{k+3/4})\right) \notag \\
    &= \mathrm{d} q^{k+1} \wedge  \mathrm{d} p^{k+1/2} - \frac{\dt}{2} \mathrm{d} q^{k+1} \wedge \nabla^2 V(q^{k+1}) \mathrm{d} q^{k+1} +\mathrm{d} q^{k+1} \wedge \mathrm{d} (\nabla \xi(q^{k+1}) \lambda^{k+3/4}) \label{eq:firstwedge}
\end{align}
As in the proof by Ref.~\onlinecite{leimkuhler_symplectic_1994}, we make use of the following two lemmata.
\begin{lem}
    Let $\mathrm{d} u$ be an arbitrary differential in $\mathbb{R}^n$ and let $A$ be any $n\times n$ real symmetric matrix. Then \[\mathrm{d} u \wedge (A\mathrm{d} u) = 0.\]
    \label{lem:symm_wedge}
\end{lem}
\begin{proof}
    Writing this expression component-wise, we have
    \[
    \mathrm{d} u \wedge (A\mathrm{d} u) = \sum_{i,j=1}^n \mathrm{d} u_i \wedge (A_{ij}\mathrm{d} u_j) =  \sum_{i,j=1}^n A_{ij}\mathrm{d} u_i \wedge \mathrm{d} u_j = - \sum_{i,j=1}^n A_{ji} \mathrm{d} u_j \wedge \mathrm{d} u_i = -\mathrm{d} u \wedge (A\mathrm{d} u)
    \]
    by the symmetry of $A$ and the antisymmetry of the wedge-product, and thus $\mathrm{d} u \wedge (A\mathrm{d} u)=0$.
\end{proof}
\begin{lem}
    Let $\mathrm{d} u$ and $\mathrm{d} v$ be arbitrary differentials in $\mathbb{R}^n$ and $\mathbb{R}^m$, respectively, and let $B$ be any $n\times m$ real matrix, then \[\mathrm{d} u \wedge (B\mathrm{d} v) = (B^\top \mathrm{d} u) \wedge \mathrm{d} v.\]
    \label{lem:skewsymm_wedge}
\end{lem}
\begin{proof}
    Writing this expression component-wise, we have
    \[
    \mathrm{d} u \wedge (B\mathrm{d} v) = \sum_{i=1}^n \sum_{j=1}^m \mathrm{d} u_i \wedge (B_{ij}\mathrm{d} v_j) = \sum_{i=1}^n \sum_{j=1}^m (B_{ij} \mathrm{d} u_i) \wedge \mathrm{d} v_j = \sum_{i=1}^n \sum_{j=1}^m (B^\top_{ji} \mathrm{d} u_i) \wedge \mathrm{d} v_j = (B^\top \mathrm{d} u) \wedge \mathrm{d} v,
    \]
    which is the desired result. 
\end{proof}
Since the Hessian is symmetric, the second term in \eqref{eq:firstwedge} vanishes by \cref{lem:symm_wedge}. The third term in \eqref{eq:firstwedge} also vanishes:
\begin{align*}
    \mathrm{d} q^{k+1} \wedge \mathrm{d} (\nabla \xi(q^{k+1}) \lambda^{k+3/4}) = \mathrm{d} q^{k+1} \wedge \nabla \xi(q^{k+1}) \mathrm{d} \lambda^{k+3/4} + \sum_{i=1}^{\ell} \lambda^{k+3/4}_i  \mathrm{d} q^{k+1} \wedge (\nabla^2\xi_i)(q^{k+1}) \mathrm{d} q^{k+1} = 0,
\end{align*}
where $\xi_i$ is the $i$-th component of the CV. Indeed, thanks to \cref{lem:skewsymm_wedge}, we have $\mathrm{d} q^{k+1} \wedge \nabla \xi(q^{k+1}) \mathrm{d} \lambda^{k+3/4} = (\nabla \xi(q^{k+1}))^\top \mathrm{d} q^{k+1} \wedge  \mathrm{d} \lambda^{k+3/4}$ which is zero due to \eqref{eq:cq_diff}. The final term equally vanishes since the Hessian is symmetric and one can thus apply \cref{lem:symm_wedge} to each term of the sum. 

We therefore have that
\begin{equation*}
    \mathrm{d} q^{k+1} \wedge \mathrm{d} p^{k+3/4} = \mathrm{d} q^{k+1} \wedge  \mathrm{d} p^{k+1/2}.
\end{equation*}
Using \eqref{eq:diff_qkplus1}, we can continue (with the same reasoning as before)
\begin{align*}
    \mathrm{d} q^{k+1} \wedge  \mathrm{d} p^{k+1/2} &= \mathrm{d} q^{k} \wedge  \mathrm{d} p^{k+1/2} + \dt \ M^{-1} \mathrm{d} p^{k+1/2} \wedge  \mathrm{d} p^{k+1/2} \\
    &= \mathrm{d} q^{k} \wedge  \left(\mathrm{d} p^{k+1/4} - \frac{\dt}{2} \mathrm{d} (\nabla V (q^{k})) + \mathrm{d} (\nabla \xi(q^k) \lambda^{k+1/2})\right) \\
    &= \mathrm{d} q^{k} \wedge \mathrm{d} p^{k+1/4} - \frac{\dt}{2} \mathrm{d} q^{k} \wedge \mathrm{d} (\nabla V (q^{k})) + \mathrm{d} q^{k} \wedge \mathrm{d} (\nabla \xi(q^k) \lambda^{k+1/2}) \\
    &= \mathrm{d} q^{k} \wedge \mathrm{d} p^{k+1/4} - \frac{\dt}{2} \mathrm{d} q^{k} \wedge \nabla^2V(q^k) \mathrm{d} q^k + \mathrm{d} q^{k} \wedge \left(\nabla \xi(q^k) \mathrm{d} \lambda^{k+1/2} + \sum_{i=1}^\ell \lambda_i^{k+1/2} (\nabla^2\xi_i)(q^k) \mathrm{d} q^k \right) \\
    &= \mathrm{d} q^{k} \wedge \mathrm{d} p^{k+1/4} - \frac{\dt}{2} \mathrm{d} q^{k} \wedge \nabla^2V(q^k) \mathrm{d} q^k + (\nabla \xi(q^k))^\top \mathrm{d} q^{k} \wedge  \mathrm{d} \lambda^{k+1/2} + \sum_{i=1}^\ell \lambda_i^{k+1/2}  \mathrm{d} q^{k} \wedge (\nabla^2\xi_i)(q^k) \mathrm{d} q^k \\
    &= \mathrm{d} q^{k} \wedge \mathrm{d} p^{k+1/4}.
\end{align*}
Following exactly the same arguments, one can equally show that $\mathrm{d}q^{{\rm b},k'}\wedge  \mathrm{d} p^{{\rm b},k'+1/4} = \mathrm{d} q^{{\rm b},k'+1} \wedge \mathrm{d} p^{{\rm b},k'+3/4}$. This concludes the proof of \cref{lem:flow_symplectic}.

We now make more precise the reversibility property we require between the forward and backward flow, already mentioned in \cref{rem:existence_lagrange_multip}. In all of the following, we assume that the time-step $\Delta t$ is sufficiently small so that, along the trajectory, the forward and backward flows $\Phi^k$ and $\Phi^{{\rm b}, k'}$ introduced in \eqref{eq:fwd_flow} and \eqref{eq:bckw_flow} and associated with the deterministic parts~\eqref{eq:Verletconstswitched} and~\eqref{eq:Verletconstswitchedbck} satisfy the following property: For all the configurations which are visited and all $k\in \{1,\ldots,K_T\}$ the following pointwise equality holds:
\begin{align}
    \Phi^{{\rm b}, K_T-k} \circ \Phi^{k-1} = \operatorname{Id}.
    \label{eq:rev_forwbackw_flows}
\end{align}
See e.g. Sect.~VII.1.4 of Ref.~\onlinecite{hairer-lubich-wanner-06} and Sect.~5.4.3 of Ref.~\onlinecite{lelievre_langevin_2012}, as well as Lemma~2.9 in Ref.~\onlinecite{lelievre-rousset-stoltz-19} where a precise setting for which this is satisfied is exhibited. Potential difficulties around this could also be dealt with using so-called reversibility checks, but we leave this for future work.

\subsubsection{Jarzynski--Crooks equality}
We are now in position to state the Jarzynski--Crooks formula. We state it here in a slightly more general form than in Ref.~\onlinecite{lelievre_langevin_2012}, but otherwise follow exactly the presentation there.
\begin{theorem}
\label{theo:JarzCrooks}
Consider the distribution $\mu_{\manq_{\xi,v_\xi}(z(t),v_z(t))}$ 
  and its normalizing constant $\mathcal{Z}_{z(t),v_z(t)}$ defined in~\eqref{eq:constrainedBoltzmann}.
  Denote by $\left((q^k,p^k, p^{k+1/4}, p^{k+3/4})_{0\leq k\leq K_T-1}, q^{K_T}, p^{K_T}\right)$ the solution of the forward discretized Langevin dynamics~\eqref{eq:flucdissconstjarz1}-\eqref{eq:flucdissconstjarz2} with initial conditions
  distributed according to
  \begin{equation}
  \label{eq:ICf_discrete}
  (q^0,p^0) \sim \mu_{\manq_{\xi,v_\xi} \pare{ z(t_0), v_z(t_0)} }(\mathrm{d}q\, \mathrm{d}p ),
  \end{equation}
  and by $\left((q^{\mathrm{b},k},p^{\mathrm{b},k}, p^{\mathrm{b},k+1/4}, p^{\mathrm{b},k+3/4})_{0\leq k \leq K_T-1}, q^{\mathrm{b},K_T}, p^{\mathrm{b},K_T}\right)$ the solution of the discretized 
  backward Langevin dynamics~\eqref{eq:flucdissconstjarzbck1}-\eqref{eq:flucdissconstjarzbck2} with initial conditions distributed according to
  \begin{equation}\label{eq:ICb_discrete}
  (q^{ { \rm b},0},p^{ {\rm b},0}) \sim \mu_{\manq_{\xi,v_\xi} \pare{ z(t_{K_T}), v_z(t_{K_T})  } }(\mathrm{d}q\, \mathrm{d}p ).
  \end{equation}
  Then, the following Jarzynski--Crooks identity holds: for any bounded 
  discrete path functional $\ph_{[0,K_T]}$, 
\begin{equation}
\label{eq:Jarz_Crooks_generalized}
\frac{\mathcal{Z}_{z(K_T), v_z(t_{K_T})   }}{\mathcal{Z}_{z(t_0), v_z(t_{0})  }  } 
    = \frac{\E \pare{ \ph_{[0,K_T]}\left((q^k,p^k, p^{k+1/4}, p^{k+3/4})_{0 \leq k \leq K_T-1}, q^{K_T}, p^{K_T}\right) \,  
      {\rm e}^{-\beta\W^{K_T} } }}{\E \pare{ \ph^{\rm r}_{[0,K_T]}\left((q^{{\rm b},k},p^{{\rm b},k}, p^{{\rm b},k+1/4}, p^{{\rm b},k+3/4})_{0 \leq k \leq K_T-1}, q^{{\rm b},K_T}, p^{{\rm b},K_T}\right) }},
\end{equation}
where $\W^{K_T}$ is computed according to~\eqref{eq:work_definition} with $\widetilde{H}=H$, and 
  $(\, \cdot \,)^{\rm r}$ denotes the composition with 
  the operation of time reversal of paths:
  \begin{align}
    \label{eq:time_reversal_on_paths_discrete}
    \ph_{[0,K_T]}^{\rm r}&\left((q^{\mathrm{b}, k}, p^{\mathrm{b},k}, p^{\mathrm{b}, k+1/4}, p^{\mathrm{b}, k+3/4})_{0\leq k\leq K_T-1}, q^{\mathrm{b}, K_T}, p^{\mathrm{b}, K_T}\right) \notag  \\
    &= \ph_{[0,K_T]}\left((q^{\mathrm{b},K_T-k},p^{\mathrm{b},K_T-k}, p^{\mathrm{b},K_T-k-1/4}, p^{\mathrm{b},K_T-k-3/4})_{0\leq k\leq K_T-1}, q^{\mathrm{b},0}, p^{\mathrm{b},0}\right).
  \end{align} 
\end{theorem}
\begin{proof}
    The proof is the same as the proof of Ref.~\onlinecite[Theorem 5.5]{lelievre_langevin_2012} and relies in particular on the reversibility property of the forward and backward dynamics \eqref{eq:rev_forwbackw_flows}. The only difference lies in the slightly different momentum constraint $(C_p)$ for the forward and backward dynamics. Since we show in \cref{lem:flow_symplectic} that the forward and backward deterministic flows are symplectic, the proof goes through in the same way.
\end{proof}

\subsection{Proof of reversibility}\label{sec:appendix_reversibility3}
\label{sec:proof_reversibility}

We are now in a position to prove \cref{theo:reversbility}, namely that the algorithm presented in \cref{subsec:algorithm} is reversible with respect to the target probability measure $\nu$.
Assume that $Q_n$ is distributed according to $\nu$. The aim is to prove that $(Q_n,Q_{n+1})$ has the same law as $(Q_{n+1},Q_n)$. In order to study the law of $(Q_n,Q_{n+1})$, let us consider a bounded measurable function~$\varphi:\R^d \times \R^d \to \R$.
Recall that $Z_n = \xi(Q_n)$ and $\widetilde{Z}_{n+1}= \xi(\widetilde{Q}_{n+1})$. One has:
\begin{align*}
\E(\varphi(Q_n,Q_{n+1}))
&=\E\left(\varphi(Q_n,\widetilde{Q}_{n+1}) \one_{U_{n+1}\le \exp(-\beta \mathcal{W}_{n+1})\frac{\rho(\widetilde{Z}_{n+1}, Z_n)}{\rho(Z_n,\widetilde{Z}_{n+1})}}\right)\\
&\quad + \E\left(\varphi(Q_n,Q_n)\one_{U_{n+1}> \exp(-\beta \mathcal{W}_{n+1})\frac{\rho(\widetilde{Z}_{n+1}, Z_n)}{\rho(Z_n,\widetilde{Z}_{n+1})}}\right) \\
&= \E\left(\varphi(Q_n,\widetilde{Q}_{n+1}) \left[1 \wedge \left( \exp(-\beta \mathcal{W}_{n+1})\frac{\rho(\widetilde{Z}_{n+1}, Z_n)}{\rho(Z_n,\widetilde{Z}_{n+1})}\right)\right]\right)\\
&\quad + \E\left(\varphi(Q_n, Q_n) \left[ 1-  1 \wedge \left(\exp(-\beta \mathcal{W}_{n+1})\frac{\rho(\widetilde{Z}_{n+1}, Z_n)}{\rho(Z_n,\widetilde{Z}_{n+1})}\right)\right]\right),
\end{align*}
where here and in the following, $a \wedge b = \min(a,b)$. 
The expectation value is taken over $Q_n\sim\nu$ (which fixes $Z_n$), and all the random variables used to generate $Q_{n+1}$: sampling $\widetilde{Z}_{n+1}$ from the density $\rho(Z_n, \cdot)$, following the steered dynamics to obtain $\widetilde{Q}_{n+1}$ and accepting or rejecting the final move. The expectation value over $Q_n$ can also be understood in two steps: 
\begin{itemize}
\item First, sample $Z_n\sim \xi\#\nu$, where 
$(\xi\#\nu)(\mathrm{d} z) = \exp(-\beta F(z)) \, \mathrm{d}z$
denotes the image of the probability measure $\nu$ by $\xi$ (see \eqref{eq:free_energy} for the definition of the free energy $F$). 

\item Second, sample a full state $Q_n$ from  $\nu(\mathrm{d}q \, | \, \xi(q)=Z_n)$, where $\nu(\mathrm{d}q \, | \, \xi(q)=z)$ denotes the the family of conditional measures
\[\nu(\mathrm{d}q \, | \, \xi(q)=z)=\frac{\mathcal{Z}_q^{-1}e^{-\beta V(q)} \delta_{\xi(q)-z}(\mathrm{d}q)}{\exp(-\beta F(z))},\]
where $\exp(-\beta F(z))$ is the normalizing constant, recall \eqref{eq:free_energy}.
\end{itemize}
By conditioning on the collective variables $(Z_n, \widetilde{Z}_{n+1})$, one gets
\begin{align}
  &\E\left(\varphi(Q_n,\widetilde{Q}_{n+1}) \left[1 \wedge \left( \exp(-\beta \mathcal{W}_{n+1})\frac{\rho(\widetilde{Z}_{n+1}, Z_n)}{\rho(Z_n,\widetilde{Z}_{n+1})}\right)\right]\right) \notag\\
  &=\int_{z, \widetilde z} \E\left(\varphi(Q_n,\widetilde{Q}_{n+1}) \left[ 1 \wedge \left( \exp(-\beta \mathcal{W}_{n+1})\frac{\rho(\xi(\widetilde{Q}_{n+1}), \xi(Q_n))}{\rho(\xi(Q_{n}),\xi(\widetilde{Q}_{n+1}))}\right)\right]\middle| (\xi(Q_n),\xi(\widetilde Q_{n+1}))=(z,\widetilde z) \right) \notag \\
  &\quad\quad  \times \exp(-\beta F(z)) \rho(z,\widetilde z) \, {\rm d}z \, {\rm d}\widetilde z \notag\\
  &=\int_{z, \widetilde z} \E\left(\varphi(Q_n,\widetilde{Q}_{n+1}) \left[ \rho(\xi(Q_{n}),\xi(\widetilde{Q}_{n+1})) \wedge \left( \exp(-\beta \mathcal{W}_{n+1}) \rho(\xi(\widetilde{Q}_{n+1}), \xi(Q_{n}))\right)\right]\middle| (\xi(Q_n),\xi(\widetilde Q_{n+1}))=(z,\widetilde z) \right) \notag \\
  &\quad\quad  \times \exp(-\beta F(z))  \, {\rm d}z \, {\rm d}\widetilde z.\notag
\end{align}

The conditional expectation value depends on the full steered path (and not just on the endpoints) via the work defined by \eqref{eq:work_definition}, which we can rewrite here as
\begin{equation}
\mathcal{W}\left((q^k, p^{k+1/4}, p^{k+3/4})_{0\leq k\leq K_T-1}, q^{K_T}\right) = \sum_{k=0}^{K_T-1} \left(\widetilde{H}(q^{k+1}, p^{k+3/4}) - \widetilde{H}(q^{k}, p^{k+1/4})\right).
\end{equation}
Notice that the work of a time-reversed path is the negative of the work of the path without time-reversal:
\begin{align}
    &\mathcal{W}\left((q^{K_T-k}, p^{K_T-k-1/4}, p^{K_T-k-3/4})_{0\leq k\leq K_T-1}, q^{0}\right) \notag\\
    &= \sum_{k=0}^{K_T-1}\left(\widetilde{H}(q^{K_T-k-1}, p^{K_T-k-3/4}) - \widetilde{H}(q^{K_T-k}, p^{K_T-k-1/4})\right) \notag\\
    &= \sum_{k=0}^{K_T-1}\left(\widetilde{H}(q^{k}, p^{k+1/4}) - \widetilde{H}(q^{k+1}, p^{k+3/4})\right) \notag\\
    &= - \mathcal{W}\left((q^{k}, p^{k+1/4}, p^{k+3/4})_{0\leq k\leq K_T-1}, q^{K_T}\right). \label{eq:work_rev_symmetry}
\end{align}
Using the generalized Jarzynski--Crooks identity, one gets, introducing the notation $\left((q^k, p^{k+1/4}, p^{k+3/4})_{0\leq k\leq K_T-1}, q^{K_T}\right)$ for the forward process and $\left((q^{\mathrm{b},k}, p^{\mathrm{b},k+1/4}, p^{\mathrm{b},k+3/4})_{0\leq k\leq K_T-1}, q^{\mathrm{b},K_T}\right)$ for the backward process as in \cref{theo:JarzCrooks},
\begin{align}
  &\E\Bigl(\varphi(Q_n,\widetilde{Q}_{n+1}) 
  \left[ \rho(\xi(Q_n),\xi(\widetilde{Q}_{n+1})) \wedge \left(\exp(-\beta \mathcal{W}_{n+1}) \rho(\xi(\widetilde{Q}_{n+1}), \xi(Q_{n}))\right)\right] 
  \Big|  (\xi(Q_n),\xi(\widetilde Q_{n+1}))=(z,\widetilde z) \Bigr)\notag\\
  &=\E\Bigl(\varphi(q^0,q^{K_T}) \Big[ \rho(\xi(q^0),\xi(q^{K_T})) \notag\\
  &\qquad \qquad \wedge \left( \exp\left(-\beta \mathcal{W}\left((q^k, p^{k+1/4}, p^{k+3/4})_{0\leq k\leq K_T-1}, q^{K_T}\right)\right) \rho(\xi(q^{K_T}),\xi(q^0))\right)\Big] \Big|  (\xi(q^0),\xi(q^{K_T}))=(z,\widetilde z)\Bigr)\notag\\
  &=\E\Bigl(\varphi(q^0,q^{K_T}) \Big[\left(\exp\left(\beta \mathcal{W}\left((q^k, p^{k+1/4}, p^{k+3/4})_{0\leq k\leq K_T-1}, q^{K_T}\right)\right)  \rho(\xi(q^0),\xi(q^{K_T})) \right)\notag\\
  &\qquad \qquad  \wedge  \rho(\xi(q^{K_T}), \xi(q^0))\Big] \exp\left(-\beta \mathcal{W}\left((q^k, p^{k+1/4}, p^{k+3/4})_{0\leq k\leq K_T-1}, q^{K_T}\right)\right)\Big|  (\xi(q^0),\xi(q^{K_T}))=(z,\widetilde z) \Bigr)\notag\\
   &\overset{\eqref{eq:Jarz_Crooks_generalized}}{=}\E\Bigl(\varphi(q^{\mathrm{b},K_T},q^{\mathrm{b},0}) \Big[\left(\exp\left(\beta \mathcal{W}\left((q^{\mathrm{b},K_T-k}, p^{\mathrm{b},K_T-k-1/4}, p^{\mathrm{b},K_T-k-3/4})_{0\leq k\leq K_T-1}, q^{\mathrm{b},0}\right)\right)  \rho(\xi(q^{\mathrm{b},K_T}),\xi(q^{\mathrm{b},0})) \right)  \notag\\
  &\qquad \qquad \wedge\, \rho(\xi(q^{\mathrm{b},0}), \xi(q^{\mathrm{b},K_T}))\Big] \Big|  (\xi(q^{\mathrm{b},0}), \xi(q^{\mathrm{b},K_T}))=(\widetilde z, z) \Bigr)\frac{\widetilde{\mathcal{Z}}(\widetilde z, v_z(t_{K_T}))}{\widetilde{\mathcal{Z}}(z, v_z(t_0))}\notag\\
   &\overset{\eqref{eq:free_energy_relations}}{=}\E\Bigl(\varphi(q^{\mathrm{b},K_T},q^{\mathrm{b},0}) \Big[\left(\exp\left(\beta \mathcal{W}\left((q^{\mathrm{b},K_T-k}, p^{\mathrm{b},K_T-k-1/4}, p^{\mathrm{b},K_T-k-3/4})_{0\leq k\leq K_T-1}, q^{\mathrm{b},0}\right)\right)  \rho(\xi(q^{\mathrm{b},K_T}),\xi(q^{\mathrm{b},0})) \right)\notag\\
  &\qquad \qquad \wedge\,  \rho(\xi(q^{\mathrm{b},0}), \xi(q^{\mathrm{b},K_T}))\Big]  \Big| (\xi(q^{\mathrm{b},0}), \xi(q^{\mathrm{b},K_T}))=(\widetilde z, z)\Bigr)\exp(-\beta (F(\widetilde z)-F(z)))\notag\\
    &\overset{\eqref{eq:work_rev_symmetry}}{=}\E\Bigl(\varphi(q^{\mathrm{b},K_T},q^{\mathrm{b},0}) \Big[\left(\exp\left(-\beta \mathcal{W}\left((q^{\mathrm{b},k}, p^{\mathrm{b},k+1/4}, p^{\mathrm{b},k+3/4})_{0\leq k\leq K_T-1}, q^{\mathrm{b},K_T}\right) \right)  \rho(\xi(q^{\mathrm{b},K_T}),\xi(q^{\mathrm{b},0})) \right)\notag\\
  &\qquad \qquad \wedge \, \rho(\xi(q^{\mathrm{b},0}), \xi(q^{\mathrm{b},K_T}))\Big]  \Big|  (\xi(q^{\mathrm{b},0}), \xi(q^{\mathrm{b},K_T}))=(\widetilde z, z)\Bigr)\exp(-\beta (F(\widetilde z)-F(z))).  \label{eq:proof_eq1}
\end{align}
In the last but one step, we used the assumption~\eqref{eq:zerovelocity_condition} that the initial and final velocity is zero, $v_z(t_0)=v_z(t_{K_T})=0$, and applied the identity \eqref{eq:free_energy_relations}.
For the final step, we used the symmetry property \eqref{eq:work_rev_symmetry} of the work function.
\begin{rem}
    As mentioned in~\cref{rem:free_energy_relations_nofixman}, when $G_M$ does not depend on $q$, there is no Fixman term to be considered and  \eqref{eq:free_energyrelation_nofixman} holds. One can use this equation instead of \eqref{eq:free_energy_relations} in the computation above, and this in turn implies that reversibility can be proven only assuming that the initial and final velocities are equal i.e. $v_z(t_0)=v_z(t_{K_T})$, but not necessarily zero (see condition~(\ref{eq:zerovelocity_condition})).
    \label{rem:proof_no_fixma_velocities}
\end{rem}

We can use \cref{lem:forw_backw_same_law} to rewrite the latter expectation value over a backward path as one over a forward path as
\begin{align}
 &\E\Bigl(\varphi(q^{\mathrm{b},K_T},q^{\mathrm{b},0}) \Big[\left(\exp\left(-\beta \mathcal{W}\left((q^{\mathrm{b},k}, p^{\mathrm{b},k+1/4}, p^{\mathrm{b},k+3/4})_{0\leq k\leq K_T-1}, q^{\mathrm{b},K_T}\right) \right)  \rho(\xi(q^{\mathrm{b},K_T}),\xi(q^{\mathrm{b},0})) \right) \notag\\
  &\qquad \qquad \wedge \,  \rho(\xi(q^{\mathrm{b},0}), \xi(q^{\mathrm{b},K_T}))\Big] \Big|  (\xi(q^{\mathrm{b},0}), \xi(q^{\mathrm{b},K_T}))=(\widetilde z, z) \Bigr)\exp(-\beta (F(\widetilde z)-F(z))) \notag\\
  &=  \E\Bigl(\varphi(q^{K_T},q^{0}) \Big[\left(\exp\left(-\beta \mathcal{W}\left((q^{k}, -p^{k+1/4}, -p^{k+3/4})_{0\leq k\leq K_T-1}, q^{K_T}\right) \right)  \rho(\xi(q^{K_T}),\xi(q^{0})) \right)\notag\\
  &\qquad \qquad \wedge \, \rho(\xi(q^{0}), \xi(q^{K_T}))\Big] \Big|  (\xi(q^{0}), \xi(q^{K_T}))=(\widetilde z, z) \Bigr)\exp(-\beta (F(\widetilde z)-F(z))) \notag\\
  &=  \E\Bigl(\varphi(q^{K_T},q^{0}) \Big[\left(\exp\left(-\beta \mathcal{W}\left((q^{k}, p^{k+1/4}, p^{k+3/4})_{0\leq k\leq K_T-1}, q^{K_T}\right) \right)  \rho(\xi(q^{K_T}),\xi(q^{0}))\right)\notag\\
  &\qquad \qquad \wedge \, \rho(\xi(q^{0}), \xi(q^{K_T}))\Big] \Big|  (\xi(q^{0}), \xi(q^{K_T}))=(\widetilde z, z) \Bigr)\exp(-\beta (F(\widetilde z)-F(z))).\label{eq:proof_eq2}
\end{align}
In the last step, we used that the work function $\mathcal{W}$ is symmetric with respect to momentum reversal since $\widetilde{H}(q,p)=\widetilde{H}(q,-p)$. 
Using \eqref{eq:proof_eq1} and \eqref{eq:proof_eq2}, we get (by just swapping the indices $z$ and $\widetilde{z}$ in the third equality)
\begin{align*}
  &\E\left(\varphi(Q_n,\widetilde{Q}_{n+1}) \left[1 \wedge \left( \exp(-\beta \mathcal{W}_{n+1})\frac{\rho(\widetilde{Z}_{n+1}, Z_n)}{\rho(Z_n,\widetilde{Z}_{n+1})}\right)\right]\right) \\
  &=\int_{z, \widetilde z} \E\Bigl(\varphi(q^{K_T},q^{0}) \Big[\left(\exp\left(-\beta \mathcal{W}\left((q^{k}, p^{k+1/4}, p^{k+3/4})_{0\leq k\leq K_T-1}, q^{K_T}\right) \right)  \rho(\xi(q^{K_T}), \xi(q^{0})) \right)   \\
  &\qquad \qquad \wedge \, \rho(\xi(q^{0}), \xi(q^{K_T}))\Big] \Big|  (\xi(q^{0}), \xi(q^{K_T}))=(\widetilde z, z) \Bigr) \exp(-\beta F(\widetilde{z}))  \, {\rm d}z \, {\rm d}\widetilde z.  \\
  &=\int_{z, \widetilde z} \E\Bigl(\varphi(q^{K_T},q^{0}) \left[1 \wedge  \left(\exp\left(-\beta \mathcal{W}\left((q^{k}, p^{k+1/4}, p^{k+3/4})_{0\leq k\leq K_T-1}, q^{K_T}\right) \right)  \frac{\rho(\xi(q^{K_T}), \xi(q^{0}))}{\rho(\xi(q^{0}), \xi(q^{K_T}))} \right) \right]   \\
  &\qquad \qquad \Big|  (\xi(q^{0}), \xi(q^{K_T}))=(\widetilde z, z) \Bigr) \rho(\widetilde z, z) \exp(-\beta F(\widetilde{z}))  \, {\rm d}z \, {\rm d}\widetilde z  \\
  &=\int_{z, \widetilde z} \E\Bigl(\varphi(q^{K_T},q^{0}) \left[1 \wedge  \left(\exp\left(-\beta \mathcal{W}\left((q^{k}, p^{k+1/4}, p^{k+3/4})_{0\leq k\leq K_T-1}, q^{K_T}\right) \right)  \frac{\rho(\xi(q^{K_T}), \xi(q^{0}))}{\rho(\xi(q^{0}), \xi(q^{K_T}))} \right) \right]  \\
  &\qquad \qquad \Big|  (\xi(q^{0}), \xi(q^{K_T}))=(z, \widetilde z) \Bigr) \rho(z, \widetilde z) \exp(-\beta F(z))  \, {\rm d}z \, {\rm d}\widetilde z.  \\
  & =\E\left(\varphi(\widetilde{Q}_{n+1}, Q_n) \left[1 \wedge \left( \exp(-\beta \mathcal{W}_{n+1})\frac{\rho(\widetilde{Z}_{n+1}, Z_n)}{\rho(Z_n,\widetilde{Z}_{n+1})}\right)\right]\right).  
\end{align*}
Finally,
    \begin{align*}
\E(\varphi(Q_n,Q_{n+1}))
&= \E\left(\varphi(\widetilde{Q}_{n+1},Q_n) \left[1 \wedge \left( \exp(-\beta \mathcal{W}_{n+1}) \frac{\rho(\widetilde{Z}_{n+1}, Z_n)}{\rho(Z_n,\widetilde{Z}_{n+1})}\right)\right]\right)\\
&\quad + \E\left(\varphi(Q_n, Q_n) \left[ 1-  1 \wedge \left(\exp(-\beta \mathcal{W}_{n+1})\frac{\rho(\widetilde{Z}_{n+1}, Z_n)}{\rho(Z_n,\widetilde{Z}_{n+1})}\right)\right]\right)\\
&=\E(\varphi(Q_{n+1},Q_n)),
\end{align*}
which concludes the proof of \cref{theo:reversbility}.
\section{Normalization of the steered dynamics \texorpdfstring{\eqref{eq:flucdissconstjarz1}-\eqref{eq:flucdissconstjarz2}}{ }}
\label{app:parameterization_scalar}
We explain in this section how the parameterization of the dynamics \eqref{eq:flucdissconstjarz1}-\eqref{eq:flucdissconstjarz2} can be simplified under certain assumptions.
We assume here that the mass matrix $M=M\mathrm{Id}$ is isotropic and furthermore write the velocity schedule as $v_z(t_k) = \widetilde{v}_z(t_k)/\Delta t$, with a velocity $\widetilde{v}_z(t_k)$ normalized by the timestep. The parameterization then simplifies considerably. With the dimensionless momentum $\widetilde{p}=\sqrt{\beta/M} p$, which is of identity covariance under the equilibrium distribution, the steered dynamics \eqref{eq:flucdissconstjarz1}-\eqref{eq:flucdissconstjarz2} on $(q,\widetilde{p})$ read
\begin{align}
  &         \dps    \widetilde{p}^{k+1/4} =   \widetilde{p}^{k} -\frac{\dt\gamma}{4M} P_M(q^k)P_M(q^k)^\top(\widetilde{p}^{k+1/4}+\widetilde{p}^{k}) 
       + \sqrt{ \frac{\dt\gamma}{2M}} P_M(q^k) {\mathcal G}^k \\[6pt]
  &\begin{cases}
    \dps  \widetilde{p}^{k+1/2} = \widetilde{p}^{k+1/4} - \displaystyle{\frac{\dt}{2}\sqrt{\frac{1}{\beta M}} \nabla (\beta \widetilde{V} (q^{k}))} + \nabla \xi(q^k) \widetilde{\lambda}^{k+1/2} &\\[6pt]
    \dps  q^{k+1} = q^{k} + \frac{\dt}{\sqrt{\beta M}} \widetilde{p}^{k+1/2}  &\\[6pt]
    \dps   \xi(q^{k+1})  = z(t_{k+1}) &(C_q) \\[6pt]
    \dps  \widetilde{p}^{k+3/4} = \widetilde{p}^{k+1/2} - \displaystyle{\frac{\dt}{2}\sqrt{\frac{1}{\beta M}} \nabla (\beta \widetilde{V} (q^{k+1}))} +\nabla \xi(q^{k+1}) \widetilde{\lambda}^{k+3/4}&\\[6pt]
    \dps   \nabla \xi (q^{k+1})^{T}  \widetilde{p}^{k+3/4} = \frac{\sqrt{\beta M}}{\Delta t}\widetilde{v}_z(t_{k+1})
    &(C_p)
\end{cases}\\
  &    \dps    \widetilde{p}^{k+1} = \widetilde{p}^{k+3/4} -\frac{\dt\gamma}{4M} P_M(q^{k+1})P_M(q^{k+1})^\top (\widetilde{p}^{k+3/4}+\widetilde{p}^{k+1})  + \sqrt{\frac{\dt\gamma}{2M}} P_M(q^{k+1}) {\mathcal G}^{k+1/2},
\end{align}
with renormalized Lagrange multipliers $\widetilde{\lambda}^{k+1/2}$ and $\widetilde{\lambda}^{k+3/4}$ that ensure that the constraints $(C_q)$ and $(C_p)$ are satisfied.
Here, we have deliberately moved $\beta$ into the gradients of the potential since it is the quantity $\beta V$ that defines the target distribution at equilibrium (note that the correction $\beta V_\mathrm{fix}$ is independent of $\beta$). With this, \textbf{Step 2} of the algorithm in \cref{subsubsec:algo} is fully defined by the schedule $z(t_k)$ and $\widetilde{v}_z(t_k)$ as well as the two parameters $\alpha_1 = \Delta t\gamma/(4M)$ and $\alpha_2 = \Delta t^2 /(\beta M)$. Note that while the Fixman term in general depends on the mass $M$, under the assumption of the isotropic mass matrix, the gradient $\nabla V_\mathrm{fix}(q)$ does not.

\section{Simplification for a linear CV}
\label{sec:linearCV}
To explicitly connect the presented algorithm to the previous one from Ref.~\onlinecite{SCHONLE2025113806}, consider the simplest case where the reaction coordinate is just given by the first $\ell$ degrees of freedom, so that we can write $q = (q_\mathrm{CV}, q_\perp)$ with $q_\mathrm{CV} = \xi(q) = (q_1,\dots,q_\ell)$ and $q_\perp = (q_{\ell+1}, \dots, q_d)$. Similarly, one can then write the momenta as $p=(p_\mathrm{CV}, p_\perp)$. 
Assume once again that 
$M=M\operatorname{Id}_d$ (with a slight abuse of notation). Then the Gram tensor is $G_M(q)=M^{-1}\operatorname{Id}_\ell$, and the projection operator for the momenta is just
\begin{align*}
    P_M(q) = \begin{pmatrix}
\mathbf{0}_{\ell\times\ell} & 0 \\
0 & \one_{(d-\ell)\times(d-\ell)}
\end{pmatrix} = P_M(q)^\top
\end{align*}
setting the momenta $p_\mathrm{CV}$ associated with the CV to zero. With the scalar parameters $\gamma$ and $\sigma$, one then has $\gamma_P(q) = P_M(q) \gamma P_M(q)^\top = \gamma P_M(q)$ and $\sigma_P(q)=P_M(q)\sigma$. Explicitly solving for the Lagrange multipliers, the resulting dynamics read (using the notation $\nabla_\perp V$ for the partial derivative of $V$ with respect to $q_\perp$)
\begin{align}
  &  \begin{cases}
       \dps    p_\mathrm{CV}^{k+1/4} =   p_\mathrm{CV}^{k}  \\
       \dps p_\perp^{k+1/4} =   p_\perp^{k} -\frac{\dt}{4M} \gamma(p_\perp^{k+1/4}+p_\perp^{k}) 
       + \sqrt{ \frac{\dt}{2}} \sigma {\mathcal G}_\perp^k
     \end{cases}'
     \label{eq:AlolinCVStep1}\\[6pt]
  &\begin{cases}
    \dps  p_\perp^{k+1/2} = p_\perp^{k+1/4} - \displaystyle{\frac{\dt}{2} \nabla_\perp V (q^{k})} &\\[6pt]
    \dps   q_\mathrm{CV}^{k+1}  = z(t_{k+1}) & \\[6pt]
    \dps  q_\perp^{k+1} = q_\perp^{k} + \frac{\dt}{M} \ p_\perp^{k+1/2}  &\\[6pt]
    \dps  p_\mathrm{CV}^{k+3/4} = M v_z(t_{k+1})&\\[6pt]    
    \dps  p_\perp^{k+3/4} = p_\perp^{k+1/2} - \displaystyle{\frac{\dt}{2} \nabla_\perp V (q^{k+1})}&
   \end{cases}\label{eq:AlolinCVStep
   2}\\
  &   \begin{cases}  
  \dps    p_\mathrm{CV}^{k+1} = p_\mathrm{CV}^{k+3/4}, \\
  \dps    p_\perp^{k+1} = p_\perp^{k+3/4} -\frac{\dt}{4M} \gamma (p_\perp^{k+3/4}+p_\perp^{k+1}) + \sqrt{\frac{\dt}{2}} \sigma {\mathcal G}_\perp^{k+1/2}. 
      \end{cases} \label{eq:AlolinCVStep3}
\end{align}
This corresponds to the underdamped Langevin dynamics on $(q_\perp, p_\perp)$ with a time-varying potential $V_k(q^k_\perp)$, which becomes obvious by rewriting $V(q^k) = V(q^k_\mathrm{CV}, q^k_\perp) \equiv V_k(q^k_\perp)$, since the CV-space coordinates $q_\mathrm{CV}$ are fully determined by the schedule. One might therefore also see the CV coordinates $q_\mathrm{CV}^k$ here as a time-varying alchemical parameter.

In the following we consider two limiting cases, the overdamped limit and the deterministic limit. 
\subsection{Overdamped limit}
\label{subsec:limit_overdamped}
For the choice $\gamma=2$ and $M=\frac{\Delta t}{2}$ (and therefore $\sigma=\frac{2}{\sqrt{\beta}}$), the forward dynamics \eqref{eq:AlolinCVStep1}-\eqref{eq:AlolinCVStep3} simplify to overdamped Langevin dynamics
\begin{align}
  &\begin{cases}
    \dps   q_\mathrm{CV}^{k+1}  = z(t_{k+1}) & \\[6pt]
    \dps  q_\perp^{k+1} = q_\perp^{k} + \sqrt{ \frac{2\dt}{\beta}} {\mathcal G}_\perp^k - \displaystyle{\dt \nabla_\perp V (q^{k})}.
   \end{cases}
\end{align}
In terms of the parameters of the normalized algorithm introduced in \cref{subsec:Parameterization}, this parameter choice corresponds to setting $\alpha_1=1$.
Using \eqref{eq:work_definition}, the associated work is calculated as
\begin{align*}
    \mathcal{W}^{K_T} = V(q^{K_T}) - V(q^{0}) + \frac{1}{4 \dt}\sum_{k=0}^{K_T-1} \left( \left\|q_\perp^{k} - q_\perp^{k+1} + \displaystyle{\dt \nabla_\perp V (q^{k+1})}\right\|_2^2 - \left\|q_\perp^{k+1} - q_\perp^{k} + \displaystyle{\dt \nabla_\perp V (q^{k})}\right\|_2^2\right).
\end{align*}
For the overall accept/reject probability in the MCMC algorithm, this gives 
\begin{align*}
    \exp\left(-\beta\mathcal{W}^{K_T}\right) = \dps \frac{e^{-\beta V(q^{K_T})}}{e^{-\beta V(q^{0})}} \frac{\exp\left(-\frac{\beta}{4\dt}\sum_{k=0}^{K_T-1}\left\|q_\perp^{k} - q_\perp^{k+1} + \displaystyle{\dt \nabla_\perp V (q^{k+1})}\right\|_2^2\right)}{\exp\left(-\frac{\beta}{4\dt}\sum_{k=0}^{K_T-1}\left\|q_\perp^{k+1} - q_\perp^{k} + \displaystyle{\dt \nabla_\perp V (q^{k})}\right\|_2^2\right)}.
\end{align*}
This algorithm is very closely related to the \textbf{asymmetric} algorithm from Ref.~\onlinecite{SCHONLE2025113806}, with the only difference that forward and backward path have swapped their definitions. 
\subsection{Deterministic limit}
\label{subsec:limit_determ}
In the deterministic limit, we have $\gamma=\sigma=0$. The forward dynamics then simplify to
\begin{align}
  &\begin{cases}
    \dps   q_\mathrm{CV}^{k+1}  = z(t_{k+1}) & \\[6pt]
    \dps  q_\perp^{k+1} = q_\perp^{k} + \frac{\dt}{M} \ p_\perp^{k} - \displaystyle{\frac{\dt^2}{2M} \nabla_\perp V (q^{k})}  &\\[6pt]
    \dps  p_\mathrm{CV}^{k+1} = M v_z(t_{k+1})&\\[6pt]    
    \dps  p_\perp^{k+1} = p_\perp^{k} - \frac{\dt}{2} \left(\nabla_\perp V (q^{k})+\nabla_\perp V (q^{k+1})\right) .&
   \end{cases}
\end{align}
This corresponds to the deterministic dynamics for the CV coordinate and Verlet integration for the orthogonal coordinates. The associated work is simply
\begin{align*}
    \W^{K_T} &= \sum_{k=0}^{K_T-1} \left(H(q^{k+1},p^{k+1}) - H(q^{k},p^{k})\right) \\
    &= H(q^{K_T},p^{K_T}) - H(q^0,p^0).
\end{align*}
For the overall MCMC algorithm defined earlier, this means that the proposal $\widetilde{Q}_{n+1} = q^{K_T}$ is accepted with probability the minimum of 1 and
\begin{align*}
    \frac{e^{-\beta H(q^{K_T},p^{K_T})}}{e^{-\beta H(q^{0},p^{0})}} \frac{\rho(\widetilde{Z}_{n+1}, Z_n)}{\rho(Z_n,\widetilde{Z}_{n+1})} = \frac{e^{-\beta V(q^{K_T})}}{e^{-\beta V(q^{0})}} \frac{e^{-\beta \|p^{K_T}\|^2/(2M)}}{e^{-\beta \|p^{0}\|^2/(2M)}}\frac{\rho(\widetilde{Z}_{n+1}, Z_n)}{\rho(Z_n,\widetilde{Z}_{n+1})} = \frac{e^{-\beta V(q^{K_T})}}{e^{-\beta V(q^{0})}} \frac{e^{-\beta \|p_\perp^{K_T}\|^2/(2M)}}{e^{-\beta \|p_\perp^{0}\|^2/(2M)}}\frac{\rho(\widetilde{Z}_{n+1}, Z_n)}{\rho(Z_n,\widetilde{Z}_{n+1})}
\end{align*}
The last equality, where the norm of the momentum is only considered on the orthogonal degrees of freedom, comes from the assumption that the initial and final velocities in CV space are equal. This is the relaxed version of \eqref{eq:zerovelocity_condition} when there is no Fixman term to be considered, see \cref{rem:constant_fixman}.

\section{Details of the wavelet basis for the \texorpdfstring{$\phi^4$}{Phi4} model}
\label{sec:wavelet_details}
As mentioned in \cref{subsec:phi4}, we use a basis of Haar-wavelets to represent the state of the $\phi^4$ model in our simulation. For an introduction to wavelet transformations we refer the reader to Refs.~\onlinecite{chui1992,mallat_wavelet_2008}, for example. Here, we only give a very practical definition. 
For a one-dimensional input signal $\{\phi_i\}_{i=1}^N$ of size $N=2^n$, set $\varphi^{(0)}_i = \phi_i$ for $i \in \{1,\ldots N\}$ and repeatedly apply the following transformation $n$ times for $j \in \{1,\ldots n-1\}$:
\begin{align*}
    \varphi^{(j+1)}_i = \frac{\varphi^{(j)}_{2i+1} + \varphi^{(j)}_{2i}}{\sqrt{2}},\quad 
    \overline{\varphi}^{(j+1)}_i = \frac{\varphi^{(j)}_{2i+1} - \varphi^{(j)}_{2i}}{\sqrt{2}}
\end{align*}
with $i \in \{1,\ldots 2^{n-j}\}$. 
This transformation constitutes a bijective linear map between the input $\{\phi_i\}_{i=1}^N$ and the wavelet fields $\varphi^{(n)}$ and $\{ \overline{\varphi}^{j}\}_{j=1}^n$. The coarsest-scale wavelet field has a single component entry:
\begin{align*}
  \varphi^{(n)}= \frac{1}{(\sqrt{2})^n}\sum_{i=1}^N \phi_i = \frac{1}{\sqrt{N}}\sum_{i=1}^N \phi_i.
\end{align*}
For the $\phi^4$ model, since $\varphi^{(n)}=\sqrt{N}\cdot \overline{\phi}$, this corresponds to the magnetization $\overline{\phi}$ up to a prefactor. In practice, we therefore run the dynamics \eqref{eq:AlolinCVStep1}-\eqref{eq:AlolinCVStep3} with $q_\mathrm{CV}=\overline{\phi}$ and $q_\perp$ the remaining wavelet components $\{ \overline{\varphi}^{j}\}_{j=1}^n$. 
\section{Details of the dimer model}
\subsection{Computation of the Fixman term}
\label{subsec:dimer_fixman}
Recall that we use the reaction coordinate defined in \eqref{eq:dimer_reaction_coordinate}:
\begin{align*}
    \xi(q) = \frac{|q_1-q_2| - r_0}{2w}.
\end{align*}
We write the gradient of the reaction coordinate $\nabla \xi(q)\in \mathbb{R}^{d\times 1}$ in the form
\begin{align*}
    \nabla \xi(q) = \frac{1}{2w}\left( \begin{array}{c}
    \dps \frac{q_1-q_2}{|q_1-q_2|} \\ [10pt]
    \dps -\frac{q_1-q_2}{|q_1-q_2|} \\
    0 \\
    \vdots \\
    0
\end{array} \right) = \frac{1}{2w} \left( \begin{array}{c}
    e_{12} \\
    -e_{12} \\
    0 \\
    \vdots \\
    0
\end{array} \right)
\end{align*}
with the normalized vector $e_{12} =  \frac{q_1-q_2}{|q_1-q_2|}\in \mathbb{R}^2$.
The Gram tensor is then given by
\[
G_M(q) = \nabla\xi(q)^\top M^{-1} \nabla\xi(q) = \frac{1}{M}|\nabla \xi(q)|^2  = \frac{1}{2 M w^2},
\]
leading to a Fixman term (defined in \eqref{eq:deffixman}) that is constant:
\[
V_\mathrm{fix}(q) = \frac{1}{2\beta} \log \det G_M(q) = -\frac{1}{2\beta} \log(2w^2 M).
\]
\subsection{Lagrange multiplier for position constraint}
\label{subsec:lagrane_multi_dimer}
From \eqref{eq:Verletconstswitched}, the enforcement of the position constraint can be rewritten as
\begin{equation*}
\left\{
\begin{aligned}
    \dps q^{k+1} &= \tilde{q}^{k+1} + \dt \ M^{-1} \nabla \xi(q^k) \lambda^{k+1/2}, \\
    \dps   \xi(q^{k+1})  &= z(t_{k+1}), \quad\quad (C_q)
\end{aligned}
\right.
\end{equation*}
with $\tilde{q}^{k+1} = q^{k} +  \dt \ M^{-1} p^{k+1/4} -  \ M^{-1} \displaystyle{\frac{\dt^2}{2} \nabla \widetilde V (q^{k})}$.
Inserting the first equation into the second and using the definition 
\eqref{eq:dimer_reaction_coordinate} of $\xi$ leads to a quadratic equation for $\lambda^{k+1/2}$ with the solution
\begin{align*}
    \lambda^{k+1/2}_\pm = \frac{wM}{\dt} \left(-e_{12}^k \cdot (\tilde{q}^{k+1}_1-\tilde{q}^{k+1}_2) \pm \sqrt{\Delta} \right),
\end{align*}
where $e_{12}^k=\frac{q_1^k-q_2^k}{|q_1^k-q_2^k|}\in \mathbb{R}^2$ is a normalized vector and the discriminant is given by
\[
\Delta = \left[e_{12}^k \cdot (\tilde{q}^{k+1}_1-\tilde{q}^{k+1}_2)\right]^2 - \left[|\tilde{q}^{k+1}_1-\tilde{q}^{k+1}_2|^2 - (2wz(t_{k+1}) + r_0)^2\right].
\]
We choose the smallest value for $\lambda^{k+1/2}$, which corresponds to choosing $\lambda^{k+1/2}_+$ if $e_{12}^k \cdot (\tilde{q}^{k+1}_1-\tilde{q}^{k+1}_2) > 0$ and $\lambda^{k+1/2}_-$ in the opposite case. This choice is consistent with ensuring that the value of the Lagrange multiplier is of order $\Delta t$ for small timesteps and ensures that \eqref{eq:rev_forwbackw_flows} is satisfied.
\subsection{Free Dimer not interacting with solvent particles}
\label{appendix:dimer_free}
Considering the dimer model and the collective variable introduced in \cref{subsec:dimer}, the free energy can be computed analytically for the special case when the solvent particles do not interact with the dimer.
We have 
\begin{align*}
    F(z) &= -\frac{1}{\beta} \log \int e^{-\beta V(q)} \delta_{\xi(q)-z}(\mathrm d q) \\
    &= -\frac{1}{\beta} \log \int e^{-\beta V(q)} |\nabla\xi(q)|^{-1} \sigma^M_{\Sigma(z)}(\mathrm d q) \\
    &=-\frac{1}{\beta} \log \left(e^{-\beta V_D(2w z + r_0)} (2w z + r_0)\right) + \mathrm{const.} \\
    &= V_D(2w z + r_0) - \frac{1}{\beta}\log(2w z + r_0) + \mathrm{const.}
\end{align*}
where we ignored constants that just shift the free energy. For the second equality, consider that if $q_1=0$ is fixed, the values of $q_2$ are constrained to a circle of radius $r=2w z + r_0$, which introduces a Jacobian term.
\section{Details of the polymer model}
\label{sec:AppPolymer}
We simulate the polymer system introduced in \cref{subsec:polymer} at a temperature $1/\beta=1.3$ using the jaxMD package\cite{jaxmd2020}. The system is contained in a three-dimensional box of length $L=7.47$ with periodic boundaries, with $N_\mathrm{solvent}=200$ solvent particles, which yields a solvent density of $N_\mathrm{solvent}/L^3=0.48$ particles per unit volume. The parameters of the polymer potentials are $d_1=1$, $d_2=1.63$ (to fix the bonding angle at 109\degree), $k_1=k_2=500$, $d_e=4.62$, $a_e=0.8$ and $k_e=9.1$, corresponding to an energy barrier of $7 \beta^{-1}$. The Lennard-Jones potential is essentially given by
\begin{align*}
    V_{\mathrm{LJ}} = 4\epsilon \left[\left(\frac{r^*}{r}\right)^{12} - \left(\frac{r^*}{r}\right)^6\right],
\end{align*}
with $r^*=\epsilon=1$, and thus very similar to the WCA potential used for the dimer. In practice, however, we use a modified version given by jaxMD that smoothly decays to zero between an onset radius $r_{\mathrm{onset}}=2.0$ and a cutoff radius $r_\mathrm{cutoff}=2.5$.

To implement \textbf{Step 1} of the algorithm in \cref{subsec:algorithm}, we use an independent proposal sampler in CV space for both choices of CV considered here. For the first CV, i.e. the one-dimensional end-to-end distance, the proposal sampler is given by a Gaussian mixture model with two modes of equal weight and width $\sigma=0.17$ at positions $z=3.53$ and $z=5.65$.
For the second, 27-dimensional CV, i.e. the Cartesian coordinates of the polymer beads, we use a normalizing flow to propose distances between the polymer beads. These distances are the eight nearest neighbor and seven second-nearest neighbor distances as well as the six distances between the first bead and beads number 4 to 9, adding up to 21 distances in total. To complete such a proposed set of distances to a full proposal of polymer coordinates, one needs to additionally specify the global translation and rotation, as well as six sign values $\eta_i\in\{\pm 1\}, \,i\in\{1,\ldots 6\}$ (explained in the following). The global translation is fixed by always positioning the central bead of the polymer at the center of the box (which we do throughout the simulation) and the global rotation is sampled uniformly. The necessity of the six sign values $\eta_i$ can be understood from the triangulation process used to determine the Cartesian polymer bead positions. Assuming one has already placed three beads with the desired distances, there will be two possible solutions (if there is a solution) to place the fourth bead at the right distance to the previous three (two intersection points of three spheres), which corresponds to choosing the sign of the dihedral angle. For a given set of 21 distances as in our case, there are six such sign values to choose, and we sample them uniformly at random. We compute the Jacobian of this reconstruction using automatic differentiation to arrive at a probability measure in the CV space of polymer positions $z\in \mathbb{R}^{27}$. The flow will sometimes produce sets of distances for which no solution to the triangulation problem exists and these samples are then directly rejected. Since this restriction only changes the normalization of the flow proposal measure, it does not enter into the overall formula of the accept/reject probability in \textbf{Step 3} of the algorithm in \cref{subsec:algorithm}. 
The flow is implemented using the FlowJax package\cite{ward2023flowjax}, with one affine whitening layer and 20 blocks of RationalQuadraticSpline transformations\cite{durkan2019neuralsplineflows} with 32 knots on the interval $[-4,4]$ and internal neural networks of depth 2 and width 30 (see the reference for further explanations). It was trained on a large dataset of one million samples obtained from a long Langevin simulation over 200 epochs with a batch size of 40,000 and a learning rate of $10^{-3}$ using the Adam optimizer with otherwise default parameters. 
\section{Additional numerical results}
\subsection{Gaussian tunnel}
\label{app:subsec:gaussiantunnel}
To illustrate the optimal choice of parameters when running the algorithm on the example of the Gaussian tunnel from \cref{subsec:gauss_tunnel}, we show the inverse mode jump cost (introduced in the main text) for different parameter values in \cref{fig:tunnel_optim}. The optimal value of $\alpha_2$ is roughly the same for all choices of $\alpha_1$, whereas the optimal number of intermediate steps $K_T$ increases with $\alpha_1$. Clearly, the optimal performance is observed for the deterministic dynamics, namely $\alpha_1=0$.
\begin{figure}[hb]
    \centering
    \includegraphics{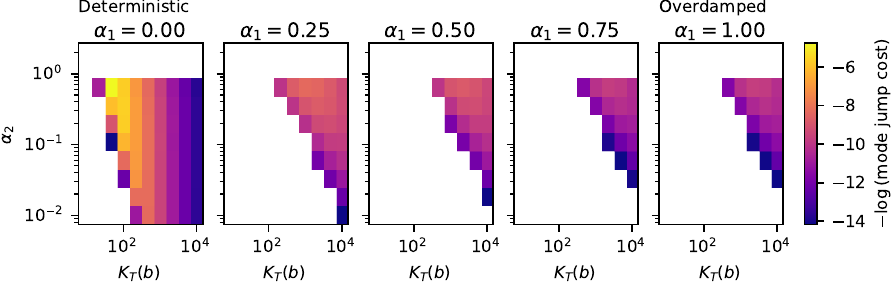}
    \caption{Inverse mode jump cost when running the algorithm on the example of the Gaussian tunnel introduced in \cref{subsec:gauss_tunnel} for different parameter values $\alpha_1$, $\alpha_2$ (introduced in \cref{subsec:Parameterization}). The number $K_T(b)$ represents the number of steps for a jump of distance $b$ and is related to the inverse of the dimensionless velocity $\widetilde{v}$ via $K_T(b) = \mathrm{ceil}(b/\widetilde{v})$. White space indicates that at least one chain never switched mode.}
    \label{fig:tunnel_optim}
\end{figure}

For the deterministic case, there is a simple continuous-time limit of $\Delta t \longrightarrow 0$, in our parameterization corresponding to $\alpha_2 \longrightarrow 0$, where the only relevant quantity determining the acceptance rate is the normalized transition time $T/\sqrt{\beta M}$. This is indeed what we observe in \cref{fig:dumbell_determ_transition_time}, where for small enough $\alpha_2$, the acceptance rate only depends on this physical time scale. For numerical performance, however, the relevant quantity is of course the acceptance rate normalized by the numerical cost.

\begin{figure}[ht]
    \centering
    \includegraphics{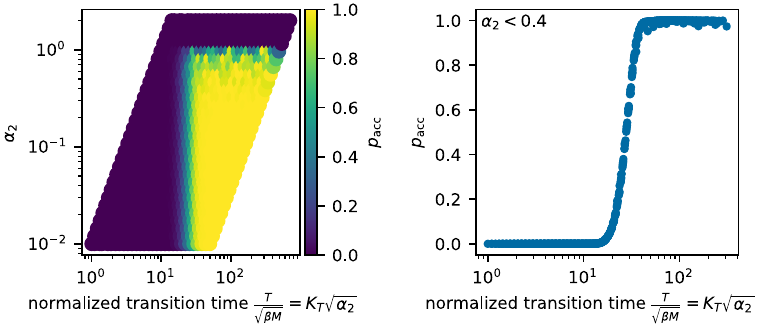}
    \caption{Acceptance rate for a fixed jump from $z=0$ to $\widetilde{z}=b$ for the Gaussian tunnel with the deterministic dynamics $\alpha_1=0$. The results shown here are the same as in the top left panel of 
    \cref{fig:Gaussians} with a different choice of axes. The right panel shows all simulations with $\alpha_2<0.4$ combined in one plot.}
    \label{fig:dumbell_determ_transition_time}
\end{figure}

\subsection{Gaussian tunnel with a non-linear collective variable}
\label{subsec:AppNonlinGauss}
We consider here a version of the Gaussian tunnel introduced in \cref{subsec:gauss_tunnel} in dimension $d=10$ with a non-linear collective variable. With the same underlying probability distribution, we consider the CV $\xi(z) = \tanh\left(\frac{z}{b}\right)\cdot \frac{b}{\tanh(1)}$. Using the push-forward of the probability density $\nu_\mathrm{CV}(z)$ under this transformation as the proposal distribution (and thus exactly the marginal of the target distribution), we show in \cref{fig:biased_simulation} how our algorithm with the schedule \eqref{eq:schedule_cos} leads to unbiased results, whereas a linear schedule that violates \cref{eq:zerovelocity_condition} or the omission of the Fixman term both lead to a bias. To test an inappropriate schedule, we simply use one with constant velocity, i.e. \eqref{eq:schedule_function_z} and \eqref{eq:schedule_function_vz} with $f(\tau)=\tau$. The Fixman term is given by $V_\mathrm{fix}(z,x^\perp)=\frac{1}{\beta}\log(\xi'(z))$ and just depends on $z$. We observe that upon omission of the Fixman term, the resulting distribution is indeed biased: its marginal along $\xi$ is proportional to $\nu_\mathrm{CV}(z)\cdot e^{\beta V_\mathrm{fix}(z)}$.
\begin{figure}[hb]
    \centering
    \includegraphics{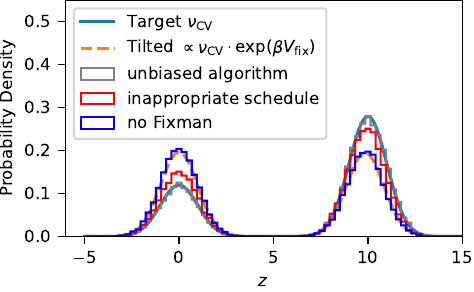}
    \caption{Simulation of the Gaussian tunnel example introduced in \cref{subsec:gauss_tunnel} (with smaller dimension $d=10$) using steered moves in a non-linear CV, $\xi(z) = \tanh\left(\frac{z}{b}\right)\cdot \frac{b}{\tanh(1)}$. We show the results for an unbiased algorithm that follows the method specified in \cref{sec:sampling_algo}, including the Fixman term \eqref{eq:deffixman} and using the cosine-schedule \eqref{eq:schedule_cos}. For comparison, we also show results for modified algorithms where (i) a linear schedule with constant velocity was used (violating \cref{eq:zerovelocity_condition}) and (ii) the Fixman term was omitted. The proposal distribution in CV-space is exactly the target distribution, histograms are shown over 20 chains and 10,000 iterations with $\alpha_1=0$, $\alpha_2=0.67$ and inverse velocity $1/\widetilde{v}=2.5$, corresponding to $K=25$ steps for a jump between the two mode centers.}
    \label{fig:biased_simulation}
\end{figure}
\clearpage
\subsection{\texorpdfstring{$\phi^4$}{Phi4} Model}
\label{subsec:AppPhi4}
To complement the results shown in \cref{fig:phi_comparison} we ran the algorithm on the $\phi^4$ model from \cref{subsec:phi4} for additional values of $\alpha_1$ and show the associated mode-jump cost in \cref{fig:Phi4optim}. 
\begin{figure}[h]
    \centering
    \includegraphics{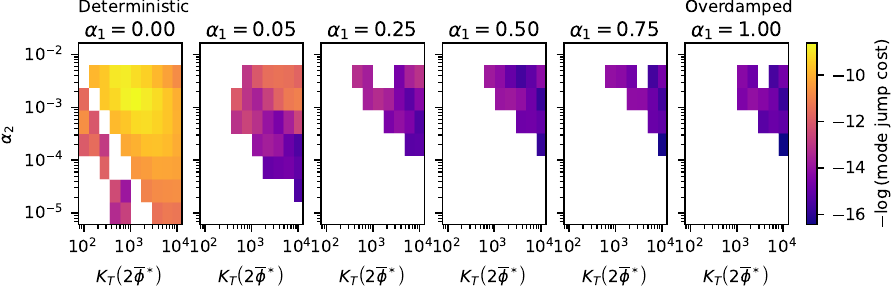}
    \caption{Mode jump cost for the $\phi^4$ model from \cref{subsec:phi4} for different parameters, estimated from running 20 MCMC chains over 10,000 iterations. The $x$-axis shows the number of steps for a steered schedule over the distance $2\overline{\phi}^*$ of the two modes}
    \label{fig:Phi4optim}
\end{figure}

\subsection{Dimer in a solvent}
\label{subsec:dimer_appendix}
In addition to \cref{fig:dimer_optimalpha0}, we ran the algorithm on the dimer model for additional values of $\alpha_1$. The associated inverse mode-jump cost is shown in \cref{fig:Dimeroptim}. Apart from the deterministic algorithm ($\alpha_1=0$), almost no mode switches were observed within the computational budget.
\begin{figure}[h]
    \centering
    \includegraphics{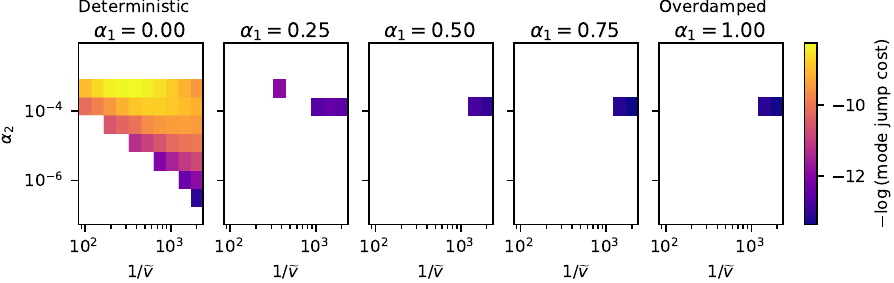}
    \caption{Mode jump cost for the Dimer model from \cref{subsec:dimer} for different parameters, estimated from running 20 MCMC chains over 5,000 iterations.  The inverse dimensionless velocity $1/\widetilde{v}$ corresponds to the number of steps for a steered schedule over a unit distance. White space indicates that at least one chain never switched mode during observation time. The left-most panel shows the same results as \cref{fig:dimer_optimalpha0} (with a different color scale).}
    \label{fig:Dimeroptim}
\end{figure}
\end{document}